\pdfoutput=1
\documentclass[final,1p,times,nopreprintline]{elsarticle}

\usepackage{amsmath,amsfonts,amssymb,amsthm,thmtools}
\usepackage{tikz}
\usetikzlibrary{shapes}

\usepackage{bm} 

\declaretheorem[style=plain,parent=section]{definition}
\declaretheorem[sibling=definition]{theorem}

\declaretheorem[sibling=definition]{proposition}
\declaretheorem[sibling=definition]{lemma}
\declaretheorem[style=remark,sibling=definition,qed={\qedsymbol}]{remark}
\declaretheorem[style=remark,sibling=definition,qed={\qedsymbol}]{example}
\declaretheoremstyle[headpunct={},notebraces={\textbf{--}~}{}]{algorithm}
\declaretheorem[style=algorithm]{problem}
\declaretheorem[style=algorithm]{algorithm}

\usepackage[shortlabels]{enumitem}
\setlist{topsep=0.25\baselineskip,partopsep=0pt,itemsep=1pt,parsep=0pt}
\usepackage{boxedminipage}

\usepackage[linkcolor=blue,colorlinks=true,citecolor=blue,urlcolor=blue,pageanchor=false]{hyperref}
\usepackage[capitalise]{cleveref}
\crefname{problem}{Problem}{Problems}
\Crefname{problem}{Problem}{Problems}

\newcommand{\storeArg}{} 

\newcommand{\bigO}[1]{\mathchoice{O\left(#1\right)}{O(#1)}{O(#1)}{O(#1)}} 
\newcommand{\softO}[1]{\mathchoice{\tilde{O}\left(#1\right)}{O\tilde{~}(#1)}{O\tilde{~}(#1)}{O\tilde{~}(#1)}} 
\newcommand{\polmultime}[1]{\mathsf{M}(#1)}
\newcommand{\expmatmul}{\omega} 
\newcommand{\algoname}[1]{{\normalfont\textsc{#1}}}
\newcommand{\algoword}[1]{\textsf{#1}}
\newcommand{\assign}{\leftarrow}
\newcommand{\comment}[1]{\texttt{\small/* #1 */}}
\newcommand{\eolcomment}[1]{\hfill\texttt{\small// #1}}
\renewcommand{\ge}{\geqslant} 
\renewcommand{\le}{\leqslant} 
\newcommand{\NN}{\mathbb{N}} 
\newcommand{\NNp}{\mathbb{N}_{> 0}} 
\newcommand{\tuple}[1]{\mathbf{#1}}  
\newcommand{\nvars}{r} 
\newcommand{\var}{X} 
\newcommand{\vars}{\bm{\var}} 
\newcommand{\field}{\mathbb{K}} 
\newcommand{\ring}{\field[\vars]} 
\newcommand{\bvPolRing}{\field[X,Y]} 
\newcommand{\mvPolRing}[1]{\field[\var_1,\ldots,\var_{#1}]} 
\newcommand{\module}[1][M]{\mathcal{#1}} 
\newcommand{\nodule}[1][N]{\mathcal{#1}} 
\newcommand{\ideal}{\mathcal{I}} 
\newcommand{\rdim}{m} 
\newcommand{\sdim}{n} %
\newcommand{\matRing}[1][\rdim]{\renewcommand\storeArg{#1}\matRingAux} 
\newcommand{\matRingAux}[1][\storeArg]{\field^{\storeArg \times #1}} 
\newcommand{\card}[1]{\mathrm{Card}(#1)}
\newcommand{\disuni}{\cup}  
\newcommand{\row}[1]{\bm{#1}} 
\newcommand{\mat}[1]{\bm{#1}} 
\newcommand{\matz}{\bm{0}} 
\newcommand{\matt}[1]{\bm{\bar{#1}}} 
\newcommand{\trsp}[1]{#1^\mathsf{T}} 
\newcommand{\matrows}[2]{\operatorname{Rows}(#1,#2)} 
\newcommand{\idMat}[1]{\mat{I}_{#1}} 
\newcommand{\rank}[1]{\mathrm{rank}(#1)}
\newcommand{\matRank}{\rho} 
\newcommand{\rkprof}{\matRank} 
\DeclareBoldMathCommand{\rkprofs}{\matRank} 
\newcommand{\svdots}{\raisebox{3pt}{$\scalebox{.75}{$\vdots$}$}} 
\newcommand{\vsdim}{D} 
\newcommand{\dimvs}{\Delta} 
\newcommand{\mul}{\cdot} 
\newcommand{\mulmat}[1]{\mat{M}_{#1}} 
\newcommand{\mulmats}{\mat{M}} 
\newcommand{\evMat}{\mat{F}} 
\newcommand{\evRow}{\row{f}} 
\newcommand{\expansion}[2][\ord,\maxDegs]{\mathcal{E}_{#1}(#2)} 
\newcommand{\contraction}[2][\ord,\maxDegs]{\mathcal{C}_{#1}(#2)} 
\newcommand{\krylov}[2][(\mulmats,\evMat)]{\mathcal{K}_{#2}#1} 
\newcommand{\basMat}{\mat{B}} 
\newcommand{\termMat}{\mat{T}} 
\newcommand{\nfMat}{\mat{N}} 
\newcommand{\ordIndex}[1][\ord,\maxDegs]{\phi_{#1}}
\newcommand{\evSpace}[2][\rdim]{\matRing[#1][#2]} 
\newcommand{\mulmatSpace}[1][\vsdim]{\matRing[#1]} 
\newcommand{\maxDeg}{\beta} 
\newcommand{\maxDegs}{\boldsymbol{\beta}} 
\newcommand{\coordVec}[1]{\row{c}_{#1}} 
\newcommand{\expnt}{e} 
\newcommand{\expnts}{\tuple{\expnt}} 
\newcommand{\ngens}{s}  
\newcommand{\gb}{\mathcal{G}} 
\newcommand{\ord}{\prec} 
\newcommand{\ordLex}{\ord_{\mathrm{lex}}} 
\newcommand{\ordDRL}{\ord_{\mathrm{drl}}} 
\newcommand{\ordTOP}[1]{\mathbin{#1^{\mathrm{top}}}} 
\newcommand{\ordPOT}[1]{\mathbin{#1^{\mathrm{pot}}}} 
\newcommand{\LM}{\mathcal{L}} 
\newcommand{\lt}[2][\ord]{\mathrm{lt}_{#1}(#2)} 
\newcommand{\ltMod}[2][\ord]{\mathrm{lt}_{#1}(#2)}  
\newcommand{\nf}[2][\ord]{\mathrm{nf}_{#1}(#2)} 
\newcommand{\monbas}{\mathcal{E}} 
\newcommand{\basVec}[1]{\varepsilon_{#1}} 
\newcommand{\monmod}{\mathcal{T}}  
\newcommand{\genBy}[1]{\langle #1 \rangle} 
\newcommand{\expSet}{\mathcal{S}}  
\newcommand{\nextExpSets}[1]{\hat{\expSet}_{#1}}
\newcommand{\border}{\mathcal{B}}  
\newcommand{\basisM}{\mathcal{F}}
\newcommand{\hyp}[1]{\mathcal{H}(#1)} 
\newcommand{\modSyz}[2][\module]{\operatorname{Syz}_{#1}(#2)} 
\newcommand{\modRel}{\operatorname{Syz}_{\mulmats}(\evMat)} 
\newcommand{\rel}[1][p]{\row{#1}} 
\newcommand{\relSpace}[1][\rdim]{\field[\vars]^{#1}} 
\newcommand{\nbvec}{t} 

\newcommand{\argfig}[1]{\begin{figure}[#1]} 

\newlist{algosteps}{enumerate}{3}
\crefname{algostepsi}{Step}{Steps}
\crefname{algostepsi}{Step}{Steps}
\Crefname{algostepsii}{Step}{Steps}
\Crefname{algostepsii}{Step}{Steps}
\Crefname{algostepsiii}{Step}{Steps}
\Crefname{algostepsiii}{Step}{Steps}

\newenvironment{algobox}[1][htbp]{
  \newcommand{\algoInfo}[2]{
    \begin{algorithm}[{\algoname{##1}}]
    \label{##2}
    ~ \smallskip

  }
  \newcommand{\dataInfos}[2]{
    \algoword{##1:}
      \begin{itemize}[leftmargin=0.8cm]
          ##2
      \end{itemize}}
  \newcommand{\dataInfo}[2]{
    \algoword{##1:} ##2 }
  \newcommand{\algoSteps}[1]{
    \setlist[algosteps,1]{leftmargin=0.6cm}
    \setlist[algosteps,2]{leftmargin=0.8cm}
    \setlist[algosteps,3]{leftmargin=0.4cm}
    \begin{algosteps}[label=\textbf{\arabic*.},ref=\arabic*]
      ##1
    \end{algosteps}
  }
  \expandafter\argfig\expandafter{#1}
  \centering
  \addtolength\fboxsep{0.1cm}
  \begin{boxedminipage}{0.99\textwidth}
  }
  {
  \end{algorithm}
  \end{boxedminipage}
  \end{figure}
}

\newenvironment{problembox}[1][htbp]{
  \newcommand{\problemInfo}[2]{
    \begin{problem}[{name=\emph{##1}}]
      \label{##2}
      ~ \smallskip

  }
  \newcommand{\dataInfos}[2]{
    \algoword{##1:}
      \begin{itemize}[leftmargin=0.8cm]
          ##2
      \end{itemize}}
  \newcommand{\dataInfo}[2]{
    \algoword{##1:} ##2 }
  \expandafter\argfig\expandafter{#1}
  \centering
  \addtolength\fboxsep{0.1cm}
  \begin{boxedminipage}{0.73\textwidth}
  }
  {
  \end{problem}
  \end{boxedminipage}
  \end{figure}
}

\begin{document}

\begin{frontmatter}

  \title{Computing syzygies in finite dimension using fast linear algebra}

  \author{Vincent Neiger}
  \address{Univ.~Limoges, CNRS, XLIM, UMR 7252, F-87000 Limoges, France}

  \author{\'Eric Schost}
  \address{University of Waterloo, Waterloo ON, Canada}

  \begin{abstract}
    We consider the computation of syzygies of multivariate polynomials in a
    finite-dimensional setting: for a $\mathbb{K}[X_1,\dots,X_r]$-module
    $\mathcal{M}$ of finite dimension $D$ as a $\mathbb{K}$-vector space, and
    given elements $f_1,\dots,f_m$ in $\mathcal{M}$, the problem is to compute
    syzygies between the $f_i$'s, that is, polynomials $(p_1,\dots,p_m)$ in
    $\mathbb{K}[X_1,\dots,X_r]^m$ such that $p_1 f_1 + \dots + p_m f_m = 0$ in
    $\mathcal{M}$. Assuming that the multiplication matrices of the $r$
    variables with respect to some basis of $\mathcal{M}$ are known, we give an
    algorithm which computes the reduced Gr\"obner basis of the module of these
    syzygies, for any monomial order, using $O(m D^{\omega-1} + r D^\omega
    \log(D))$ operations in the base field $\mathbb{K}$, where $\omega$ is the
    exponent of matrix multiplication. Furthermore, assuming that $\mathcal{M}$
    is itself given as $\mathcal{M} = \mathbb{K}[X_1,\dots,X_r]^n/\mathcal{N}$,
    under some assumptions on $\mathcal{N}$ we show that these multiplication
    matrices can be computed from a Gr\"obner basis of $\mathcal{N}$ within the
    same complexity bound. In particular, taking $n=1$, $m=1$ and $f_1=1$ in
    $\mathcal{M}$, this yields a change of monomial order algorithm along the
    lines of the FGLM algorithm with a complexity bound which is sub-cubic in
    $D$.
  \end{abstract}


\end{frontmatter}

\section{Introduction}
\label{sec:introduction}

In what follows, $\field$ is a field and we consider the polynomial ring $\ring
= \mvPolRing{\nvars}$. The set of \(\rdim \times \sdim\) matrices over a ring
\(\mathcal{R}\) is denoted by \(\mathcal{R}^{\rdim \times \sdim}\); when
orientation matters, a vector in \(\mathcal{R}^\sdim\) is considered as being
in \(\mathcal{R}^{1 \times \sdim}\) (row vector) or in \(\mathcal{R}^{\sdim
\times 1}\) (column vector). We are interested in the efficient computation of
relations, known as syzygies, between elements of a $\ring$-module $\module$.

Let us write the $\ring$-action on $\module$ as $(p,f) \in \ring\times \module
\mapsto p \mul f$, and let $f_1,\ldots,f_\rdim$ be in $\module$. Then, for a
given monomial order $\ord$ on $\ring^\rdim$, we want to compute the Gr\"obner
basis of the kernel of the homomorphism
\[
  \begin{array}{rcl}
    \ring^\rdim & \to & \module \\ 
    (p_1,\ldots,p_\rdim) & \mapsto & p_1\mul f_1 + \cdots + p_\rdim\mul f_\rdim.
  \end{array}
\]
This kernel is called the \emph{module of syzygies} of $(f_1,\dots,f_m)$ and
written $\modSyz{f_1,\dots,f_m}$.

In this paper, we focus on the case where $\module$ has finite dimension
$\vsdim$ as a $\field$-vector space; as a result, the quotient \(\ring^\rdim /
\modSyz{f_1,\dots,f_m}\) has dimension at most \(\vsdim\) as a
\(\field\)-vector space. Then one may adopt a linear algebra viewpoint detailed
in the next paragraph, where the elements of $\module$ are seen as row vectors
of length $\vsdim$ and the multiplication by the variables is represented by
so-called multiplication matrices. This representation was used and studied in
\cite{AuzingerStetter1988,Mourrain1999,AlonsoMarinariMora2003,KeKrRo05}, mainly
in the context where \(\module\) is a quotient \(\ring/\ideal\) for some ideal
\(\ideal\) (thus zero-dimensional of degree \(\vsdim\)) and more generally a
quotient \(\ring^\sdim / \nodule\) for some submodule
\(\nodule\subseteq\ring^\sdim\) with \(\sdim\in\NNp\) (see \cite[Sec.\,4.4
and\,6]{AlonsoMarinariMora2003}). This representation with multiplication
matrices allows one to perform computations in such a quotient via linear
algebra operations.

Assume we are given a $\field$-vector space basis $\basisM$ of $\module$. For
$i$ in $\{1,\dots,\nvars\}$, the matrix of the structure morphism $f \mapsto
\var_i \mul f$ with respect to this basis is denoted by $\mulmat{i}$; this
means that for $f$ in $\module$ represented by the vector $\evRow \in
\matRing[1][\vsdim]$ of its coefficients on $\basisM$, the coefficients of $X_i
\mul f \in \module$ on $\basisM$ are $\evRow \mulmat{i}$. We call
$\mulmat{1},\dots,\mulmat{\nvars}$ \emph{multiplication matrices}; note that
they are pairwise commuting. The data formed by these matrices defines the
module $\module$ up to isomorphism; we use it as a representation of
\(\module\). For $p$ in $\ring$ and for $f$ in $\module$ represented by the
vector $\evRow \in \matRing[1][\vsdim]$ of its coefficients on \(\basisM\), the
coefficients of $p\mul f \in \module$ on $\basisM$ are $\evRow \,
p(\mulmat{1},\ldots,\mulmat{\nvars})$; hereafter this vector is written $p \mul
\evRow$. From this point of view, syzygy modules can be described as follows.

\begin{definition}
  \label{dfn:syzygy_module}
  For \(\rdim\) and \(\vsdim\) in \(\NNp\), let $\mulmats =
  (\mulmat{1},\ldots,\mulmat{\nvars})$ be pairwise commuting matrices in
  $\matRing[\vsdim]$, and let $\evMat \in \matRing[\rdim][\vsdim]$. Denoting by
  $\row{f}_1,\ldots,\row{f}_m$ the rows of $\evMat$, for $\rel =
  (p_1,\ldots,p_\rdim) \in \relSpace$ we write
  \[
    \rel \mul \evMat
    = p_1 \mul \row{f}_1 + \cdots + p_\rdim \mul \row{f}_m
    = \row{f}_{1}\, p_1(\mulmats) + \cdots + \row{f}_{\rdim}\, p_\rdim(\mulmats) \in \matRing[1][\vsdim].
  \]
  The \emph{syzygy module} $\modRel$, whose elements are called syzygies for
  $\evMat$, is defined as
  \[
    \modRel = \{ \rel \in \relSpace \mid \rel \mul \evMat = \matz \};
  \] 
  as noted above, $\ring^\rdim/\modRel$ has dimension at most \(\vsdim\) as a
  \(\field\)-vector space.
\end{definition}

\noindent In particular, if in the above context $\evMat$ is the matrix of the
coefficients of $f_1,\dots,f_m \in \module$ on the basis $\basisM$, then
$\modRel = \modSyz{f_1,\dots,f_m}$. Our main goal in this paper is to give a
fast algorithm to solve the following problem (for the notions of monomial
orders and Gr\"obner basis for modules, we refer to \cite{Eisenbud95} and the
overview in \cref{sec:preliminaries}).

\begin{problembox}
  \problemInfo
  {Gr\"obner basis of syzygies}
  {pbm:grb}

  \dataInfos{Input}{
    \item a monomial order $\ord$ on $\ring^\rdim$,
    \item pairwise commuting matrices $\mulmats = (\mulmat{1},\ldots,\mulmat{\nvars})$ in $\matRing[\vsdim]$,
    \item a matrix $\evMat \in \evSpace{\vsdim}$. }

  \dataInfo{Output}{ the reduced $\ord$-Gr\"obner basis of $\modRel$.  }
\end{problembox}

\begin{example}
  \label{ex:hermite_pade}
  An important class of examples has $\nvars=1$; in this case, we are working
  with univariate polynomials. Restricting further, consider the case
  $\module=\field[X_1]/\langle X_1^\vsdim \rangle$ endowed with the canonical
  $\field[X_1]$-module structure; then $\mulmat{1}$ is the upper shift matrix,
  whose entries are all $0$ except those on the superdiagonal which are $1$.
  Given $f_1,\dots,f_\rdim$ in $\module$, $(p_1,\dots,p_\rdim) \in
  \field[X_1]^\rdim$ is a syzygy for $f_1,\dots,f_\rdim$ if $p_1 f_1 + \cdots +
  p_\rdim f_\rdim = 0 \bmod X_1^\vsdim$. Such syzygies are known as
  \emph{Hermite-Pad\'e approximants} of $(f_1,\dots,f_\rdim)$
  \cite{Hermite1893,Pade1894}. Using moduli other than $X_1^\vsdim$ leads one
  to generalizations such as \emph{M-Pad\'e approximants} or \emph{rational
  interpolants} (corresponding to a modulus that splits into linear factors)
  \cite{Mahler68,Beckermann92,BarBul92}.

  For $\nvars=1$, $\modRel$ is known to be free of rank $\rdim$. Bases of such
  $\field[X_1]$-modules are often described by means of their so-called
  \emph{Popov} form~\cite{Popov72,Kailath80}. In commutative algebra terms,
  this is a \emph{term over position} Gr\"obner basis. Another common choice is
  the \emph{Hermite} form, which is a \emph{position over term} Gr\"obner
  basis~\cite{KoRaTa07}.
\end{example}

\begin{example}
  \label{ex:gb_ideal}
  For arbitrary $\nvars$, let $\ideal$ be a zero-dimensional ideal in $\ring$
  and let $\module = \ring/\ideal$ with the canonical $\ring$-module structure.
  Then, taking $\rdim=1$ and $f_1=1 \in \module$, we have
  \[
    \modSyz{f_1} = \{ p \in \ring \mid p \, f_1 = 0 \} =
    \{ p \in \ring \mid p \in \ideal \} = \ideal.
  \]
  Suppose we know a Gr\"obner basis of $\ideal$ for some monomial order
  $\ord_1$, together with the corresponding monomial basis of $\module$, and
  the multiplication matrices of $X_1,\dots,X_r$ in \(\module\). Then
  solving~\cref{pbm:grb} amounts to computing the Gr\"obner basis of $\ideal$
  for the new order $\ord$.

  More generally for a given \(f_1=f \in \module\), the case \(\rdim=1\)
  corresponds to the computation of the \emph{annihilator} of \(f\) in
  \(\ring\), often denoted by \(\mathrm{Ann}_{\ring}(\{f\})\). Indeed, the
  latter set is defined as \(\{p \in \ring \mid p f = 0\}\), which is precisely
  \(\modSyz{f}\).
\end{example}

\begin{example}
 \label{eg:moller_buchberger}
  For an arbitrary $\nvars$, let $\bm{\alpha}_1,\dots,\bm{\alpha}_\vsdim$ be
  pairwise distinct points in $\field^\nvars$, with
  $\bm{\alpha}_i=(\alpha_{i,1},\dots,\alpha_{i,\nvars})$ for all $i$. Let
  $\ideal$ be the vanishing ideal of
  $\{\bm{\alpha}_1,\dots,\bm{\alpha}_\vsdim\}$ and $\module = \ring/\ideal$. As
  above, take $\rdim=1$ and $f_1=1$, so that $\modSyz{1}=\ideal$.

  The Chinese Remainder Theorem gives an explicit isomorphism $\module \to
  \field^\vsdim$ that amounts to evaluation at
  $\bm{\alpha}_1,\dots,\bm{\alpha}_\vsdim$. The multiplication matrices induced
  by this module structure on $\field^\vsdim$ are diagonal, with $\mulmat{j}$
  having diagonal $(\alpha_{1,j},\dots,\alpha_{\vsdim,j})$ for $1 \le j \le
  \nvars$. Taking $\evMat = [1~\cdots~1] \in \field^{1 \times D}$,
  solving~\cref{pbm:grb} allows us to compute the Gr\"obner basis of the
  vanishing ideal $\ideal$ for any given order $\ord$. This problem was
  introduced by M\"oller and Buchberger~\cite{MolBuc82}; it may be extended to
  cases where vanishing multiplicities are prescribed~\cite{MaMoMo93}.
\end{example}

\begin{example}
  \label{ex:matrix_moller_buchberger}
  Now we consider an extension of the M\"oller-Buchberger problem due to
  Kehrein, Kreuzer and Robbiano~\cite{KeKrRo05}. Given $\nvars$ pairwise
  commuting $d\times d$ matrices $\mat{N}_1, \dots, \mat{N}_\nvars$, we look
  for their ideal of syzygies, that is, the ideal $\ideal$ of polynomials
  $p\in\ring$ such that $p(\mat{N}_1, \dots, \mat{N}_\nvars) = \matz$. When
  $\nvars=1$, this ideal is generated by the minimal polynomial of $\mat{N}_1$.

  One may see this problem in our framework by considering $\module =
  \matRing[d][d]$ endowed with the $\ring$-module structure given by $X_k\cdot
  \mat{A} = \mat{A} \mat{N}_k$ for all $1 \le k \le \nvars$ and \(\mat{A} \in
  \matRing[d][d]\). The ideal \(\ideal\) defined above is the module of
  syzygies \(\modSyz{f}\) of the identity matrix \(f = \idMat{d} \in \module\),
  so we have $\rdim=d$ and $\vsdim=d^2$ here. To form the input of
  \cref{pbm:grb}, we choose as a basis of \(\module\) the list of elementary
  matrices \(\basisM =
  (\coordVec{1,1},\ldots,\coordVec{1,d},\ldots,\coordVec{d,1},\ldots,\coordVec{d,d})\)
  where \(\coordVec{i,j}\) is the matrix in \(\matRing[d][d]\) whose only
  nonzero entry is a \(1\) at index \((i,j)\). Then, for $k \in
  \{1,\ldots,\nvars\}$, the multiplication matrix $\mulmat{k}$ is the Kronecker
  product \(\idMat{d} \otimes \mat{N}_k\), that is, the block diagonal matrix
  in $\matRing[d^2][d^2]$ with \(d\) diagonal blocks, each of them equal to
  \(\mat{N}_k\). Besides, the input \(\mat{F} \in \matRing[1][d^2]\) is the
  vector of coordinates of \(f=\idMat{d}\) on the basis \(\basisM\), so that
  \(\ideal = \modSyz[\mulmats]{\mat{F}}\) where \(\mulmats =
  (\mulmat{1},\ldots,\mulmat{\nvars})\).
\end{example}

\begin{example}
  \label{ex:multivariate_pade}
  Our last example is a multivariate extension of Hermite-Pad\'e approximants
  and involves arbitrary parameters \(\nvars\ge 1\) and \(\rdim\ge 2\). For a
  positive integer $d$, consider the ideal
  $\ideal=\genBy{X_1^d,\ldots,X_\nvars^d}$, and let \(\module = \ring/\ideal\).
  Then for given $f_2,\ldots,f_\rdim$ in $\module$, which may be seen as
  polynomials truncated in degree $d$ in each variable, the syzygy module
  \(\modSyz{-1,f_2,\ldots,f_\rdim}\) is
  \[
    \{ (p,q_2,\ldots,q_\rdim) \in \ring^\rdim \mid
      p = q_2 f_2 + \cdots + q_\rdim f_\rdim \bmod \genBy{X_1^d,\ldots,X_\nvars^d} \}.
  \]
  It was showed in \cite[Thm.\,3.1]{Fitzpatrick97} that this module is
  generated by
  \[
    \{f_i \coordVec{1} + \coordVec{i} \mid 2 \le i \le \rdim\}
    \cup
    \{X_k^d \coordVec{i} \mid 1 \le k \le \nvars, 2 \le i \le \rdim\}
    ,
  \]
  where \(\coordVec{i}\) is the coordinate vector \((0,\ldots,0,1,0,\ldots,0)
  \in \ring^\rdim\) with $1$ at index $i$. In the same reference, two
  algorithms are given to find the Gr\"obner basis of this syzygy module for
  arbitrary monomial orders. One algorithm uses the FGLM change of order
  (extended to modules), based on the fact that the above generating set is a
  Gr\"obner basis for a well-chosen order. The other one proceeds iteratively
  on the $d^\nvars$ vanishing conditions; this approach is similar to the
  above-mentioned algorithm of \cite{MaMoMo93}, and can be seen as a
  multivariate generalization of the classical iterative algorithm for
  univariate Hermite-Pad\'e approximation \cite{BarBul92,BecLab94}.

  For the linear algebra viewpoint used in this paper, consider the
  $\field$-vector space basis $\basisM$ of \(\module\) formed by all monomials
  \(X_1^{\expnt_1} \cdots X_\nvars^{\expnt_\nvars}\) for \(0 \le \expnt_1,
  \ldots, \expnt_\nvars < d\) ordered by the lexicographic order \(\ordLex\)
  with \(X_\nvars \ordLex \cdots \ordLex X_1\). Computing bases of the above
  syzygy module is an instance of \cref{pbm:grb} with $\vsdim = d^\nvars$,
  taking for $\evMat$ the matrix of the coefficients of
  $(-1,f_2,\ldots,f_\rdim)$ on the basis $\basisM$, and for
  \(\mulmat{1},\ldots,\mulmat{\nvars}\) the multiplication matrices of
  $X_1,\ldots,X_\nvars$ in \(\module\) with respect to \(\basisM\). These
  matrices are types of block upper shift matrices which are nilpotent of order
  \(d\). Taking \(\nvars=2\) for example, \(\mulmat{1}\) is block-diagonal with
  all diagonal blocks equal to the \(d\times d\) upper shift matrix, while
  $\mulmat{2}$ is a matrix formed by $d$ rows and $d$ columns of $d \times d$
  blocks which are all zero except those on the (blockwise) superdiagonal which
  are equal to the $d\times d$ identity matrix.
\end{example}

\paragraph{\textbf{Main result}}

For $\nvars$ variables and an input module of vector space dimension $\vsdim$,
we design an algorithm whose cost is essentially that of performing fast linear
algebra operations with $\nvars$ scalar matrices of dimensions
$\vsdim\times\vsdim$. In the rest of the paper, $\expmatmul$ is a feasible
exponent for matrix multiplication over the field \(\field\); the best known
bound is $\expmatmul < 2.38$~\cite{CopWin90, LeGall14}.

\begin{theorem}
  \label{thm:grb}
  Let $\ord$ be a monomial order on $\relSpace$, let
  $\mulmat{1},\ldots,\mulmat{\nvars}$ be pairwise commuting matrices in
  $\matRing[\vsdim]$, and let $\evMat \in \matRing[\rdim][\vsdim]$. Then there
  is an algorithm which solves \Cref{pbm:grb} using
  \begin{align*}
    \bigO{\rdim \vsdim^{\expmatmul-1} +  \vsdim^\expmatmul (\nvars + \log(d_1 \cdots d_\nvars))} 
    \subset \bigO{\rdim \vsdim^{\expmatmul-1} +  \nvars \vsdim^\expmatmul \log(\vsdim)}
  \end{align*}
  operations in \(\field\), where $d_k \in \{1,\ldots,\vsdim\}$ is the degree
  of the minimal polynomial of $\mulmat{k}$, for $1 \le k \le \nvars$.
\end{theorem}

This theorem is proved in \cref{sec:syzygy}, based on
\cref{algo:syzygy_module_basis} which is built upon \cref{algo:monomial_basis}
(computing the monomial basis) and \cref{algo:normal_forms} (computing normal
forms). Commonly encountered situations involve $\rdim \le \vsdim$ (see
\cref{ex:gb_ideal,eg:moller_buchberger,ex:matrix_moller_buchberger}), in which
case the cost bound can be simplified as $\bigO{\nvars \vsdim^\expmatmul
\log(\vsdim)}$. The interest in the more precise cost bound involving
$d_1,\ldots,d_\nvars$ comes from situations such as that of
\cref{ex:multivariate_pade}, where \(d_1 \cdots d_\nvars = d^\nvars = \vsdim\)
since all the matrices $\mulmat{1},\ldots,\mulmat{\nvars}$ have a minimal
polynomial of degree \(d\); in that case, the first cost bound in the theorem
becomes $\bigO{\rdim \vsdim^{\expmatmul-1} +  \nvars \vsdim^\expmatmul +
\vsdim^\expmatmul \log(\vsdim)}$. Another refinement of the cost bound is given
in \cref{rmk:krylov_lextop} for a particular order \(\ord\), namely the term
over position lexicographic order.

Our algorithm deals with the multiplication matrices
$\mulmat{1},\ldots,\mulmat{\nvars}$ one after another, allowing us to rely on
an approach inspired by that designed for the univariate case \(\nvars=1\) in
\citep{JeNeScVi17}. This also helps us to introduce fast matrix multiplication,
by avoiding the computation of many vector-matrix products involving each time
a different matrix $\mulmat{k}$, and instead grouping these operations into
some matrix-matrix products involving $\mulmat{1}$, then some others involving
$\mulmat{2}$, etc. To our knowledge, this is the first time a general answer is
given to~\cref{pbm:grb}. For problems such as
\cref{ex:gb_ideal,eg:moller_buchberger,ex:matrix_moller_buchberger,ex:multivariate_pade},
ours is the first algorithm that relies on fast linear algebra techniques, with
the partial exception of~\cite{FaGaHuRe13}, as discussed below.

Our cost bound can be compared to the input and output size of the problem. The
input matrices are represented using $\rdim \vsdim + \nvars \vsdim^2$ field
elements, and we will see in~\cref{sec:preliminaries,sec:syzygy} that one can
represent the output Gr\"obner basis using at most $\rdim \vsdim + \nvars
\vsdim^2$ field elements as well.

\paragraph{\textbf{Overview of the algorithm}}

To introduce matrix multiplication in our solution to \cref{pbm:grb}, we rely
on a linearization into linear algebra problems over $\field$. From $\mulmats$
and $\evMat$, we build a matrix over $\field$ whose nullspace corresponds to a
set of syzygies in $\modRel$. This matrix is called a \emph{multi-Krylov
matrix}, in reference to its structure which exhibits a collection of Krylov
subspaces of $\field^\vsdim$.

The multi-Krylov matrix is a generalization to several variables and to an
arbitrary monomial order of the (striped-)Krylov matrices considered for
example in \cite{Kailath80,BecLab00}; already in \cite[Chap.\,6]{Kailath80},
Popov and Hermite bases of modules over the univariate polynomials are obtained
by means of Krylov matrix computations. Its construction is similar to the
Sylvester matrix and more generally to the Macaulay matrix
\cite{Sylvester1853,Macaulay1902,Macaulay1916}, which are commonly used when
adopting a linear algebra viewpoint while dealing with operations on univariate
and multivariate polynomials.

In what follows, given $\expnts=(\expnt_1,\dots,\expnt_\nvars)$ in
$\NN^\nvars$, we write $\vars^\expnts= X_1^{\expnt_1}\cdots
X_\nvars^{\expnt_\nvars}$ and $\mulmats^\expnts = \mulmat{1}^{\expnt_1} \cdots
\mulmat{\nvars}^{\expnt_\nvars}$; the $i$th coordinate vector in $\field^\rdim$
is written $\coordVec{i}$. Then the construction of the multi-Krylov matrix is
based on viewing the product $\vars^{\expnts} \coordVec{i} \mul \evMat$, for a
monomial $\vars^\expnts \coordVec{i} \in \relSpace$ and $\evMat$ in
$\matRing[\rdim][\vsdim]$, as $\row{f}_i \mulmats^\expnts$, where $\row{f}_i$
is the $i$th row of $\mat{F}$. Since a polynomial in $\relSpace$ is a
$\field$-linear combination of monomials, this identity means that a syzygy in
$\modRel$ may be interpreted as a $\field$-linear relation between row vectors
of the form $\row{f}_{i} \mulmats^{\expnts}$.

Choosing some degree bounds $\maxDegs=(\maxDeg_1,\dots,\maxDeg_\nvars) \in
\NNp^\nvars$, our multi-Krylov matrix is then formed by all such rows
$\row{f}_{i} \mulmats^{\expnts}$, for $1 \le i \le \rdim$ and $\tuple{0} \le
\expnts < \maxDegs$ entrywise, ordered according to the monomial order
\(\ord\). For sufficiently large $\maxDegs$ (taking $\beta_k = \vsdim$ for all
$k$ is enough, for instance), we show in \cref{sec:syzygy:monbas} that the row
rank profile of this matrix corresponds to the $\ord$-monomial basis of the
quotient $\relSpace/\modRel$.

The main task of our algorithm is to compute this row rank profile. Adapting
ideas in the algorithms of \cite{FaGiLaMo93,MaMoMo93,Fitzpatrick97}, one would
iteratively consider the rows of the matrix, looking for a linear relation with
the previous rows by Gaussian elimination. When such a linear relation is
found, the corresponding row can be discarded. Now, the multi-Krylov structure
further permits to discard all the rows that correspond to monomial multiples
of the leading term of the discovered syzygy, even before computing these rows.
At some point, the set of rows to be considered is exhausted, and we can deduce
the row rank profile.

In this approach, a row of the multi-Krylov matrix is computed by multiplying
one of the already computed rows by one of the multiplication matrices. This
results in many vector-matrix products, with possibly different matrices each
time: this is an obstacle towards incorporating fast matrix multiplication. We
circumvent this by introducing the variables one after another, thus seemingly
not respecting the order of the rows specified by the monomial order; yet, we
will manage to ensure that this order is respected in the end. When dealing
with one variable $X_k$, we process successive powers $\mulmat{k}^{2^\expnt}$
in the style of Keller-Gehrig's algorithm \cite{KelGeh85}, using a logarithmic
number of iterations. 

Finally, from the monomial basis, one can easily find the minimal generating
set of the leading module of $\modRel$. The union of the monomial basis and of
these generators is a set of monomials, which thus corresponds to a submatrix
of the multi-Krylov matrix; the left nullspace of this submatrix, computed in
reduced row echelon form, yields the reduced $\ord$-Gr\"obner basis of
syzygies.

\paragraph{\textbf{Previous work}}

An immediate remark is that the number of field entries of the multi-Krylov
matrix is $\rdim \maxDeg_1\cdots\maxDeg_\nvars\vsdim \in \bigO{\rdim
\vsdim^{\nvars+1}}$, which significantly exceeds our target cost. Exploiting
the structure of this matrix is therefore a common thread in all efficient
algorithms.

For the univariate case $\nvars=1$, first algorithms with cost quadratic in
$\vsdim$ were given in~\cite{Sergeyev87,Paszkowski87} for Hermite-Pad\'e
approximation (\cref{ex:hermite_pade}); they returned a single syzygy of small
degree. Later, work on this case \cite{BarBul91,BecLab94} showed how to compute
a basis of the module of syzygies in time $\bigO{\rdim \vsdim^2}$ and the more
general M-Pad\'e approximation was handled in the same complexity in
\cite{BarBul92,Beckermann92,BecLab97,BecLab00}, with the algorithm of
\cite{BecLab00} returning the reduced \(\ord\)-Gr\"obner basis (called the
shifted Popov form in that context) of syzygies for an arbitrary monomial order
\(\ord\). Then a cost quasi-linear in \(\vsdim\) was achieved for M-Pad\'e
approximation by a divide and conquer approach based on fast polynomial matrix
multiplication \cite{BecLab94,GiJeVi03,JeNeScVi17}, with the most recent
algorithms returning the reduced \(\ord\)-Gr\"obner basis of syzygies for an
arbitrary \(\ord\) at such a cost \cite{JeNeScVi16,JeannerodNeigerVillard2019}.
M-Pad\'e approximation exactly corresponds to instances of \cref{pbm:grb} with
$\nvars=1$ and a multiplication matrix in Jordan normal form (which is further
nilpotent for Hermite-Pad\'e approximation) \cite{JeNeScVi17}. For achieving
such costs, that are better than quadratic in \(\vsdim\), it is necessary that
the multiplication matrix exhibits such a structure and that the algorithm
takes advantage of it, since in general merely representing the input
multiplication matrix already requires \(\vsdim^2\) field elements. 

Still for \(\nvars=1\), the case of an upper triangular multiplication matrix
\(\mulmat{1}\) was handled in \cite[Algo.\,FFFG]{BecLab00} by relying on the
kind of linearization discussed above. This algorithm exploits the triangular
shape to iterate on the $\vsdim$ leading principal submatrices of
\(\mulmat{1}\), which extends the iteration for M-Pad\'e approximation in
\cite{BarBul92,Beckermann92,BecLab97} and costs \(\bigO{\rdim \vsdim^2 +
\vsdim^3}\) operations (see \cite[Prop.\,6.5]{Neiger16a} for a complexity
analysis). To take advantage of fast matrix multiplication, another approach
was designed in \cite[Sec.\,7]{JeNeScVi17}, considering the same Krylov matrix
as in \cite{BecLab00} but processing its structure in the style of
Keller-Gehrig's algorithm \cite{KelGeh85}. The algorithm in \cite{JeNeScVi17}
supports an arbitrary matrix \(\mulmat{1}\) at the cost \(\bigO{\rdim
\vsdim^{\expmatmul-1} + \vsdim^\expmatmul \log(d)}\), and the algorithm in this
paper can be seen as a generalization of it to the multivariate case since both
coincide when \(\nvars=1\) (up to the conversion between shifted Popov forms and
reduced \(\ord\)-Gr\"obner bases). This generalization is not straightforward:
besides the obvious fact that the output basis usually has more than \(\rdim\)
elements because most submodules of \(\ring^\rdim\) are not free (unlike in the
univariate case), as highlighted above the multivariate case also involves a
more complex management of the order in which rows in the multi-Krylov matrix
are inserted and processed, in relation with the monomial order \(\ord\) and
the successive introductions of the different variables.

On the other hand, previous algorithms dealing with the case $\nvars \ge 1$
were developed independently of this line of work, starting from M\"oller and
Buchberger's algorithm~\cite{MolBuc82} and Faug\`ere \emph{et al.}'s FGLM
algorithm~\cite{FaGiLaMo93}. The former computes the ideal of a finite set of
points (\cref{eg:moller_buchberger}), while the latter is specialized to the
change of monomial order for ideals (\cref{ex:gb_ideal}), with a cost bound in
$\bigO{\nvars \vsdim^3}$. Note however that the input in the FGLM algorithm is
not the same as in \cref{pbm:grb}; this is discussed below.

A first generalization of \cite{MolBuc82,FaGiLaMo93} was presented in
\cite[Algo.\,1]{MaMoMo93}, still in the case $\rdim=1$ and \(f_1=1\): the input
describes \(\module = \ring/\ideal\) for some zero-dimensional ideal \(\ideal\)
of degree \(\vsdim\), and the algorithm outputs a \(\ord\)-Gr\"obner basis of
$\ideal$. The cost bound for this algorithm involves a term $\bigO{\nvars
\vsdim^3}$, but also a term related to the description of the input. This input
description is different from ours: it consists of a set of linear functionals
which defines the ideal \(\ideal\); thus one should be careful when comparing
this work to our results. Another related extension of \cite{MolBuc82} is the
Buchberger-M\"oller algorithm for matrices given in
\cite[Sec.\,4.1.2]{KeKrRo05}, which solves \cref{ex:matrix_moller_buchberger};
the runtime is not specified in that reference.

Another type of generalizations of \cite{MolBuc82} was detailed in
\cite[Algo.\,2]{MaMoMo93}, \cite[Algo.\,4.7]{Fitzpatrick97}, and
\cite[Algo.\,3.2]{OKeeFit02}, with the last reference tackling syzygy modules
with \(\rdim\ge1\) and arbitrary \(f_1,\ldots,f_\rdim\) like in this paper, yet
with assumptions on \(\module\); the cost bounds given in \cite{MaMoMo93} and
\cite{Fitzpatrick97} involve a term in $\bigO{\nvars \vsdim^3}$ whereas
\cite{OKeeFit02} does not report on complexity bounds. The assumptions on
\(\module\) are specified in the input of \cite[Algo.\,2]{MaMoMo93}, in
\cite[Eqn.~(4.1)]{Fitzpatrick97}, and in \cite[Eqn.~(5)]{OKeeFit02}, and they
imply that one can solve the problem by finding iteratively Gr\"obner bases for
a sequence of ``approximating'' modules of syzygies which decrease towards the
actual solution. Such assumptions lead to instances of \cref{pbm:grb} which
generalize to \(\nvars\ge 1\) the above-mentioned cases of M-Pad\'e
approximation and of a triangular multiplication matrix when \(\nvars=1\);
\cite[Sec.\,5]{OKeeFit02} explicitly mentions the link with Beckermann and
Labahn's algorithm for M-Pad\'e approximation \cite{BecLab97,BecLab00}. In both
the univariate and multivariate settings, it seems that such an iterative
approach cannot be applied to the general case of \cref{pbm:grb}, where the
input module has no other property than being finite-dimensional, and is
described through arbitrary commuting matrices.

For the particular case of the change of order for ideals (\cref{ex:gb_ideal}),
so with $\rdim=1$, when the target order is the lexicographic order $\ordLex$,
and under the assumption that the ideal is in Shape Position, fast matrix
multiplication was used for the first time in \cite{FaGaHuRe14}, yielding a
sub-cubic complexity. Indeed, if for example \(X_\nvars\) is the smallest
variable, these assumptions ensure that only $\mulmat{\nvars}$ is needed; with
this matrix as input, \cite[Prop.\,3]{FaGaHuRe14} gives a probabilistic
algorithm to compute the $\ordLex$-Gr\"obner basis of syzygies within the cost
bound $\bigO{\vsdim^\expmatmul \log(\vsdim) + \nvars \polmultime{\vsdim}
\log(\vsdim)}$. Besides ideas from \cite{FauMou11,FauMou16}, this uses
repeating squaring as in \cite{KelGeh85}. In this paper, we manage to
incorporate fast matrix multiplication without assumption on the module, and
for an arbitrary order.

Still for the particular of \cref{ex:gb_ideal}, Faug\`ere and Mou give in
\cite{FauMou11,FauMou16} probabilistic algorithms based on sparse linear
algebra. These papers do not consider the computation of the multiplication
matrices, which are assumed to be known.  While we do not make any sparsity
assumption on the multiplication matrices, for the sake of comparison we still
summarize this approach below. Noticing that the multiplication matrices
arising from the context of polynomial system solving are often sparse,
\cite{FauMou11,FauMou16} tackle \cref{pbm:grb} from a point of view similar to
the Wiedemann algorithm. Evaluating the monomials in $\ring/\ideal$ at certain
linear functionals allows one to build a multi-dimensional recurrent sequence
which admits $\ideal$ as its ideal of syzygies (this is only true for some type
of ideals $\ideal$). In terms of the multi-Krylov matrices we are considering
in \cref{sec:syzygy} concerning \cref{pbm:grb}, this is similar to introducing
an additional projection on the right of the multiplication matrices to take
advantage of the sparsity by using a black-box point of view. Then recovering a
$\ordLex$-Gr\"obner basis of this ideal of syzygies can be done via the
Berlekamp-Massey-Sakata algorithm \cite{Sakata90}, or the recent improvements
in \cite{BerBoyFau15,BerBoyFau16,BerFau16,BerFau18}.

\paragraph{\textbf{Application: change of order}}

The FGLM algorithm \cite{FaGiLaMo93} solves the change of order problem for
Gr\"obner bases of ideals in \(\bigO{\nvars \vsdim^3}\) operations for
arbitrary orders \(\ord_1\) and \(\ord_2\): starting \emph{only} from a
\(\ord_1\)-Gr\"obner basis $\gb_1$ for the input order $\ord_1$, it computes
the \(\ord_2\)-Gr\"obner basis \(\gb_2\) of the ideal \(\ideal =
\genBy{\gb_1}\). Following \cite[Sec.\,2.1]{FauMou11}, one can view the
algorithm as a two-step process: it first computes from \(\gb_1\) the
multiplication matrices of \(\module=\ring/\ideal\) with respect to the
\(\ord_1\)-monomial basis, and then finds \(\gb_2\) as a set of
\(\field\)-linear relations between certain normal forms modulo $\gb_1$. The
algorithm extends to the case of submodules of \(\ring^\sdim\) for
\(\sdim\ge1\) (see e.g.~\cite[Sec.\,2]{Fitzpatrick97}).

Our algorithm for \cref{pbm:grb} incorporates fast linear algebra into the
second step, so once the multiplication matrices are known, one can find the
reduced $\ord_2$-Gr\"obner basis in $\bigO{\sdim\vsdim^{\expmatmul-1} + \nvars
\vsdim^\expmatmul \log(\vsdim)}$ operations. We now discuss how fast linear
algebra may be incorporated into the computation of the multiplication matrices
(\cref{pbm:mulmat}).

\begin{problembox}
  \problemInfo
  {Computing the multiplication matrices}
  {pbm:mulmat}

  \dataInfos{Input}{
    \item a monomial order $\ord$ on $\ring^\sdim$,
    \item a reduced $\ord$-Gr\"obner basis $\{f_1,\ldots,f_\ngens\} \subset
      \ring^\sdim$ such that $\module=\ring^\sdim/\genBy{f_1,\ldots,f_\ngens}$ has finite
      dimension as a \(\field\)-vector space. }

  \dataInfo{Output}{%
    the multiplication matrices $\mulmat{1},\ldots,\mulmat{\nvars}$ of
    $X_1,\ldots,X_\nvars$ in $\module=\ring^\sdim/\genBy{f_1,\ldots,f_\ngens}$ with
    respect to its $\ord$-monomial basis.  }
\end{problembox}

Our solution to this problem finds its roots in \cite[Sec.\,4]{FaGaHuRe14},
which focuses on the case where $\sdim=1$ and \(\ideal =
\genBy{f_1,\ldots,f_\ngens}\) is an ideal in \(\ring\). In the context studied
in this reference only the matrix \(\mulmat{\nvars}\) of the smallest variable
\(X_\nvars\) is needed, and it is showed that this matrix can be simply read
off from the input Gr\"obner basis without arithmetic operations, under some
structural assumption on the ideal of leading terms of \(\ideal\) described in
\cite[Prop.\,7]{FaGaHuRe14}.  Here we consider the more general case of
submodules \(\nodule = \genBy{f_1,\ldots,f_\ngens}\) of \(\ring^\sdim\), and we
design an algorithm which computes all multiplication matrices
\(\mulmat{1},\ldots,\mulmat{\nvars}\) in $\module=\ring^\sdim/\nodule$ using
$\bigO{\nvars \vsdim^\expmatmul \log(\vsdim)}$ operations, under a structural
assumption on the module of leading terms of \(\nodule\). Situations where this
assumption on \(\genBy{\ltMod{\nodule}}\) holds typically involve a monomial order
such that \(X_\nvars\coordVec{i} \ord \cdots \ord X_1\coordVec{i}\) for \(1\le
i\le\sdim\). The assumption, described below, naturally extends the one
from \cite{FaGaHuRe14} to the case of submodules.

\begin{definition}
  \label{dfn:structure_monomial}
  For a monomial submodule $\monmod \subseteq \ring^\sdim$, the assumption
  \(\hyp{\monmod}\) is: ``for all monomials $\mu \in \monmod$, for all \(j \in
  \{1,\ldots,\nvars\}\) such that \(X_j\) divides \(\mu\), for all
  \(i\in\{1,\ldots,j-1\}\), the monomial $\frac{X_i}{X_j} \mu$ belongs to
  $\monmod$''.
\end{definition}

\noindent In fact, instead of considering all monomials in \(\monmod\), one can
observe that \(\hyp{\monmod}\) holds if and only if the property holds for each
monomial in the minimal generating set of \(\monmod\) (see
\cref{lem:structural_assumption_mingens}).

\begin{theorem}
  \label{thm:mulmat}
  For \(\sdim\ge1\), let $\ord$ be a monomial order on $\ring^\sdim$ and let
  $\{f_1,\ldots,f_\ngens\}$ be a reduced $\ord$-Gr\"obner basis defining a
  submodule \(\nodule = \genBy{f_1,\ldots,f_\ngens}\) of $\ring^\sdim$ such
  that $\ring^\sdim/\nodule$ has dimension \(\vsdim\) as a \(\field\)-vector
  space.  Assuming $\hyp{\genBy{\ltMod{\nodule}}}$,
  \begin{itemize}
    \item \cref{pbm:mulmat} can be solved using $\bigO{\nvars \vsdim^\expmatmul
      \log(\vsdim)}$ operations in $\field$;
    \item the change of order problem, that is, computing the reduced
      \(\ord_2\)-Gr\"obner basis of \(\nodule\) for a monomial order
      \(\ord_2\), can be solved using $\bigO{\sdim\vsdim^{\expmatmul-1} +
      \nvars \vsdim^\expmatmul \log(\vsdim)}$ operations in $\field$.
  \end{itemize}
\end{theorem}

Concerning the first item, an overview of our approach is presented in
\cref{sec:mulmat:overview} and the detailed algorithms and proofs are in
\cref{sec:mulmat:krylov_evaluation,sec:mulmat:computing_mulmat}, with a
slightly refined cost bound in \cref{prop:algo:multiplication_matrices}. The
second item is proved in \cref{sec:mulmat:change_order}, based on
\cref{algo:change_order} which essentially calls our algorithms to compute
first the multiplication matrices (\cref{pbm:mulmat}) and then the
\(\ord_2\)-Gr\"obner basis of \(\nodule\) by considering a specific module of
syzygies (\cref{pbm:grb}).

Our structural assumption has been considered before, in particular in the case
where $\sdim=1$ and we work modulo an ideal $\ideal =
\genBy{f_1,\ldots,f_\ngens}$ (and the monomial order is such that \(X_\nvars
\ord \cdots \ord X_1\), which is always true up to renaming the variables).
In this case, it holds assuming for instance that the coefficients of
$f_1,\dots,f_\ngens$ are generic (that is, pairwise distinct indeterminates
over a given ground field) and that the Moreno-Socias conjecture holds
\cite[Sec.\,4.1]{FaGaHuRe14}. Another important situation where the assumption
holds is when the leading ideal \(\genBy{\lt{\ideal}}\) is \emph{Borel-fixed}
and the characteristic of $\field$ is zero, see~\cite[Sec.\,15.9]{Eisenbud95}
and~\cite[Sec.\,4.2]{FaGaHuRe14}. A theorem first proved by Galligo in power
series rings~\cite{Galligo74}, then by Bayer and
Stillman~\cite[Prop.~1]{BaSt87} for a homogeneous ideal $\ideal$ in $\ring$
shows that after a generic change of coordinates, \(\genBy{\lt{\ideal}}\) is
Borel-fixed.

The most general version of this result we are aware of is due to
Pardue~\cite{Pardue94}. It applies to $\ring$-submodules $\nodule \subset
\ring^\sdim$, for certain monomial orders $\ord$ on $\ring^\sdim$; the precise
conditions on $\ord$ are too technical to be stated here, but they hold in
particular for the term over position order induced by a monomial order on
$\ring$ which refines the (weighted) total degree. In such cases, Pardue shows
that after a generic linear change of variables, $\genBy{\ltMod{\nodule}}$
satisfies a Borel-fixedness property on $\ring^\sdim$ which implies that
$\hyp{\genBy{\ltMod{\nodule}}}$ holds, at least in characteristic zero.

For polynomial system solving, with $\sdim=1$, an interesting
particular case of the change of order problem is that of
$\ord_1=\ordDRL$ being the degree reverse lexicographic order and
$\ord_2=\ordLex$ being the lexicographic order (so that in
characteristic zero, Pardue's result shows that in generic
coordinates, our structural assumption holds for such inputs). Fast
algorithms for this case have been studied in
\cite{FaGaHuRe14,FaGaHuRe13}. The former assumes the ideal $\ideal$ is
in Shape Position, whereas this assumption is not needed here. In
\cite{FaGaHuRe13}, an algorithm is designed to compute the
multiplication matrices from a $\ordDRL$-Gr\"obner basis in
time $\bigO{\maxDeg \nvars^\expmatmul \vsdim^\expmatmul}$, where $\maxDeg$
is the maximum total degree of the elements of the input Gr\"obner
basis. This is obtained by iterating over the total degree: the normal
forms of all monomials of the same degree are dealt with using only
one call to Gaussian elimination. While this does not require an
assumption on the leading ideal, it is unclear to us how to remove the
dependency in $\maxDeg$ in general.

\paragraph{\textbf{Outline}}

\cref{sec:preliminaries} gathers preliminary material used in the rest of the
paper: some notation, as well as basic definitions and properties related to
monomial orders, Gr\"obner bases, and monomial staircases. Then
\cref{sec:syzygy} gives algorithms and proofs concerning the computation of
bases of syzygies, leading to the main result of this paper (\cref{thm:grb}),
while \cref{sec:mulmat} focuses on the computation of the multiplication
matrices (\cref{thm:mulmat}); both sections are introduced with a more detailed
outline of their content.

\section{Notations and definitions}
\label{sec:preliminaries}

\paragraph{\textbf{Monomial orders and Gr\"obner bases for modules}} Hereafter, we
consider a multivariate polynomial ring \(\ring=\mvPolRing{\nvars}\), for some
field $\field$. Recall that the coordinate vectors are denoted by
\(\coordVec{1},\ldots,\coordVec{\rdim}\), that is,
\[
  \coordVec{j} = (0,\ldots,0,1,0,\ldots,0) \in \ring^\rdim \text{ with } 1 \text{
  at index } j.
\]
A monomial in $\ring$ is defined from exponents $\expnts =
(e_1,\ldots,e_\nvars) \in \NN^\nvars$ as $\vars^\expnts = X_1^{e_1} \cdots
X_\nvars^{e_\nvars}$ and a monomial in $\ring^\rdim$ is $\mu \coordVec{j} =
(0,\ldots,0,\mu,0,\ldots,0)$ with $1 \le j \le \rdim$ and where $\mu$ is any
monomial of $\ring$. A term in $\ring$ or in $\ring^\rdim$ is a monomial
multiplied by a nonzero scalar from $\field$. Given terms $\mu$ and $\nu$ in
$\ring$, we say that the term $\mu \coordVec{j}$ is divisible by the term $\nu
\coordVec{k}$ if $j = k$ and $\mu$ is divisible by $\nu$ in $\ring$.

A submodule of \(\ring^\rdim\) generated by monomials of \(\ring^\rdim\) is
called a monomial submodule. A submodule of \(\ring^\rdim\) generated by
homogeneous polynomials of \(\ring^\rdim\) is called a homogeneous submodule.

Following \cite[Sec.\,15.3]{Eisenbud95}, a monomial order on \(\ring^\rdim\) is
a total order $\ord$ on the monomials of \(\ring^\rdim\) such that, for any
monomials \(\mu,\nu\) of \(\ring^\rdim\) and any monomial \(\kappa\) of
\(\ring\), \(\mu \ord \nu \;\;\text{implies}\;\; \mu \ord \kappa \mu \ord
\kappa \nu\). Examples of common monomial orders on $\ring^\rdim$ are so-called
\emph{term over position} (top) and \emph{position over term} (pot). In both
cases, we start from a monomial order on $\ring$ written $\ord$. Then, given
monomials $\mu \coordVec{i}$ and $\nu \coordVec{j}$, we say that $\mu
\coordVec{i} \ordTOP \ord \nu\coordVec{j}$ if $\mu \ord \nu$ or if $\mu=\nu$
and $i<j$. Similarly, we say that $\mu \coordVec{i} \ordPOT \ord
\nu\coordVec{j}$ if $i < j$ or if $i=j$ and $\mu\ord \nu$.

For a given monomial order \(\ord\) on \(\ring^\rdim\) and an element \(f \in
\ring^\rdim\), the \(\ord\)-leading term of \(f\), denoted by \(\lt{f}\), is
the term of \(f\) whose monomial is the greatest with respect to \(\ord\). This
extends to any collection \(\mathcal{F} \subseteq \ring^\rdim\) of polynomials:
\(\ltMod{\mathcal{F}}\) is the set of leading terms \(\{\lt{f} \mid f \in
\mathcal{F}\}\) of the elements of \(\mathcal{F}\). In particular, for a module
\(\nodule\) in \(\ring^\rdim\), \(\genBy{\ltMod{\nodule}}\) is a monomial
submodule of \(\ring^\rdim\) which is called the \(\ord\)-leading module of
\(\nodule\). 

\begin{definition}[Gr\"obner basis]
  \label{dfn:grobner_basis}
  Let \(\ord\) be a monomial order on \(\ring^\rdim\) and let \(\nodule\) be a
  \(\ring\)-submodule of \(\ring^\rdim\). A subset
  \(\{f_1,\ldots,f_\ngens\} \subset \nodule\) is said to be a \(\ord\)-Gr\"obner
  basis of \(\nodule\) if the \(\ord\)-leading module of \(\nodule\) is
  generated by \(\{\lt{f_1},\ldots,\lt{f_\ngens}\}\),
  i.e.~\(\genBy{\ltMod{\nodule}} = \genBy{\lt{f_1},\ldots,\lt{f_\ngens}}\).
\end{definition}

There is a specific $\ord$-Gr\"obner basis of $\nodule$, called the reduced
$\ord$-Gr\"obner basis of $\nodule$, which is uniquely defined in terms of the
module $\nodule$ and the monomial order $\ord$. Namely, this is the Gr\"obner
basis $\{f_1,\ldots,f_\ngens\}$ of $\nodule$ such that for $1 \le i\le \ngens$,
$\lt{f_i}$ is monic and does not divide any term of $f_j$ for $j\neq i$.

\paragraph{\textbf{Monomial basis and staircase monomials}} In what follows, the
submodules \(\nodule \subseteq \ring^\rdim\) we consider are such that the
quotient \(\ring^\rdim/\nodule\) has finite dimension as a \(\field\)-vector
space. We will often use its basis formed by the monomials not in
\(\ltMod{\nodule}\) \citep[Thm.\,15.3]{Eisenbud95}; this basis is denoted by
\(\monbas = (\basVec{1},\ldots,\basVec{\vsdim})\) and called the
\emph{\(\ord\)-monomial basis} of \(\ring^\rdim/\nodule\). Any polynomial $f
\in \ring^\rdim$ can be uniquely written $f = g + h$, where $g \in \nodule$ and
$h \in \ring^\rdim$ is a $\field$-linear combination of the monomials in
$\monbas$; this polynomial $h$ is called the \emph{$\ord$-normal form of $f$
(with respect to $\nodule$)} and is denoted by $\nf{f}$. We extend the notation
to sets of polynomials \(\mathcal{F} \subseteq \ring^\rdim\), that is,
\(\nf{\mathcal{F}} = \{\nf{f} \mid f \in \mathcal{F}\}\).

As in \cite[Sec.\,3]{MaMoMo93} (which focuses on the case of ideals), we will
use other sets of monomials related to this monomial basis. First, we consider
the monomials obtained by multiplying those of $\monbas$ by a variable:
\[
  \expSet = \{ \var_k \basVec{j} \mid 1 \le k \le \nvars, 1 \le j \le \vsdim \}
  \cup
  \{ \coordVec{i} \mid 1 \le i \le \rdim \text{ such that } \coordVec{i} \not\in \monbas\} .
\]
This allows us to define the border \(\border = \expSet - \monbas\), which is a
set of monomial generators of \(\genBy{\ltMod{\nodule}}\). Then the polynomials
\(\{ \mu - \nf{\mu} \mid \mu \in \border \}\) form a canonical generating set of
\(\nodule\), called the \(\ord\)-border basis of \(\nodule\)
\cite{MaMoMo91,MaMoMo93}. Finally, we consider the minimal generating set
\(\LM\) of \(\genBy{\ltMod{\nodule}}\): it is a subset of \(\border\) such that
the reduced \(\ord\)-Gr\"obner basis of \(\nodule\) is \(\gb = \{\mu - \nf{\mu}
\mid \mu \in \LM\}\). In particular, we have $\LM = \lt{\gb} = \{\lt{f}
\mid f \in \gb\}$.

By construction, \(\card{\LM} = \card{\gb} \le \card{\border} \le
\card{\expSet}\). Above, \(\expSet\) is defined as the union of a set of
cardinality at most \(\nvars \vsdim\) and a set of cardinality at most
\(\rdim\), hence \(\card{\expSet} \le \nvars\vsdim + \rdim\). Besides, since
the coordinate vectors in the second set are in the minimal generating set
\(\LM\) of \(\genBy{\ltMod{\nodule}}\), we have \(\card{\border-\LM} \le \nvars
\vsdim\).
Note that if the upper bound on \(\card{\expSet}\) is an equality, then all
coordinate vectors \(\coordVec{1},\ldots,\coordVec{\rdim}\) are in \(\LM\), which
holds only in the case \(\nodule = \ring^\rdim\) (in particular \(\monbas =
\emptyset\) and \(\vsdim=0\)).

Finally, we give a characterization of the structural assumption of monomial
submodules described in \cref{dfn:structure_monomial}, showing that one can
focus on the monomials in the minimal generating set instead of all monomials
in the module.
\begin{lemma}
  \label{lem:structural_assumption_mingens}
  Let \(\monmod\) be a monomial submodule of $\ring^\rdim$ and let
  \(\{\mu_1,\ldots,\mu_\ngens\}\) be the minimal generating set of \(\monmod\).
  Then \(\hyp{\monmod}\) holds if and only if for all \(k \in
  \{1,\ldots,\ngens\}\), for all \(j \in \{1,\ldots,\nvars\}\) such that
  \(X_j\) divides \(\mu_k\), for all \(i\in\{1,\ldots,j-1\}\), we have
  $\frac{X_i}{X_j} \mu_k \in \monmod$.
\end{lemma}
\begin{proof}
  Obviously, \(\hyp{\monmod}\) implies the latter property since each \(\mu_k\)
  is a monomial in \(\monmod\). Conversely, we assume that for all \(k \in
  \{1,\ldots,\ngens\}\), for all \(j \in \{1,\ldots,\nvars\}\) such that
  \(X_j\) divides \(\mu_k\), for all \(i\in\{1,\ldots,j-1\}\), we have
  $\frac{X_i}{X_j} \mu_k \in \monmod$, and we want to prove that
  \(\hyp{\monmod}\) holds. Let \(\mu\) be a monomial in \(\monmod\), let \(j
  \in \{1,\ldots,\nvars\}\) be such that \(X_j\) divides \(\mu\), and let
  \(i\in\{1,\ldots,j-1\}\); we want to prove that $\frac{X_i}{X_j} \mu \in
  \monmod$. Since \(\monmod\) is a monomial module, $\mu = \nu \mu_k$ for some
  monomial \(\nu\in\ring\) and \(1\le k\le\ngens\). By assumption, \(X_j\)
  divides \(\nu \mu_k\), thus either \(X_j\) divides \(\nu\) or \(X_j\) divides
  \(\mu_k\).  In the first case \(\frac{X_i}{X_j} \mu = (\frac{X_i}{X_j} \nu)
  \mu_k \in \monmod\), and in the second case $\frac{X_i}{X_j} \mu_k \in
  \monmod$ by assumption and therefore \(\frac{X_i}{X_j} \mu = \nu
  (\frac{X_i}{X_j} \mu_k) \in \monmod\).
\end{proof}

\section{Computing bases of syzygies via linear algebra}
\label{sec:syzygy}

In this section, we focus on \cref{pbm:grb} and we prove
\cref{thm:grb}. Thus, we are given pairwise commuting matrices
$\mulmats = (\mulmat{1},\ldots,\mulmat{\nvars})$ in
$\matRing[\vsdim]$, a matrix $\evMat \in \evSpace{\vsdim}$ and a
monomial order $\ord$ on $\ring^\rdim = \mvPolRing{\nvars}^\rdim$; we 
compute the reduced $\ord$-Gr\"obner basis of $\modRel$.

The basic ingredient is a \emph{linearization} of the problem, meaning
that we will interpret all operations on polynomials as operations
of \(\field\)-linear algebra. In
\cref{sec:syzygy:linearization}, we show a correspondence
between syzygies of bounded degree and vectors in the nullspace of a
matrix over \(\field\) which exhibits a structure that we call
\emph{multi-Krylov}.  The multi-Krylov matrix is formed by
multiplications of \(\evMat\) by powers of the multiplications
matrices; its rows are ordered according to the monomial order
\(\ord\) given as input of \cref{pbm:grb}.

Then, in \cref{sec:syzygy:monbas}, we show that the row rank
profile of this multi-Krylov matrix exactly corresponds to the
\(\ord\)-monomial basis of the quotient \(\relSpace/\modRel\). To
compute this row rank profile efficiently, we use in particular an
idea from \cite{KelGeh85} which we extend to our context, while also
both exploiting the structure of the multi-Krylov matrix to always work on
a small subset of its rows and ensuring that the rows are considered
in the right order. Finally, in \cref{sec:syzygy:relbas}, we
exploit the knowledge of the \(\ord\)-monomial basis to compute the
reduced \(\ord\)-Gr\"obner basis of syzygies.

\subsection{Monomial basis as the rank profile of a multi-Krylov matrix}
\label{sec:syzygy:linearization}

We first describe \emph{expansion} and \emph{contraction} operations,
which convert polynomials of bounded degrees into their coefficient
vectors and vice versa (the bound is written $\maxDegs$ below).  It
will be convenient to rely on the following indexing function.

\begin{definition}
  \label{dfn:indexing_function}
  Let \(\maxDegs = (\maxDeg_1,\ldots,\maxDeg_\nvars) \in \NNp^\nvars\) and let
  \(\ord\) be a monomial order on \(\relSpace\). Then we define the
  \emph{\((\ord,\maxDegs)\)-indexing function} \(\ordIndex\) as the unique
  bijection
  \[
    \ordIndex :
    \{ \vars^\expnts \coordVec{i} \mid \tuple{0} \le \expnts < \maxDegs, 1\le i \le \rdim\}
    \to \{1, \ldots, \rdim \maxDeg_1 \cdots \maxDeg_\nvars \}
  \]
  which is increasing for \(\ord\), that is, such that \(\vars^{\expnts} \coordVec{i} \ord
  \vars^{\expnts'} \coordVec{i'}\) if and only if \(\ordIndex(\vars^{\expnts}
  \coordVec{i}) < \ordIndex(\vars^{\expnts'} \coordVec{i'})\).
\end{definition}
In other words, take the sequence of monomials $(\vars^\expnts
\coordVec{i}, \tuple{0} \le \expnts < \maxDegs, 1\le i \le \rdim)$ and
sort it according to $\ord$; then $\ordIndex(\vars^{\expnts}
\coordVec{i})$ is the index of $\vars^{\expnts} \coordVec{i}$ in the
sorted sequence (assuming indices start at $1$).

Hereafter, \(\ring_{<\maxDegs}\) stands for the set of polynomials \(p
\in \ring\) such that \(\deg_{X_k}(p) < \maxDeg_k\) for \(1\le k\le
\nvars\). This yields a $\field$-linear correspondence between
bounded-degree polynomials and vectors
\[
\begin{array}{cccc}
\mathcal{E}_{\ord,\maxDegs}: & \ring_{<\maxDegs}^m\ & \to &  \matRing[1][\rdim\maxDeg_1\cdots\maxDeg_\nvars] \\
& \rel = \displaystyle\sum_{\substack{\row{f} = \vars^\expnts \coordVec{i} \\ \tuple{0} \le \expnts < \maxDegs, 1\le i \le \rdim}}  u_{\row{f}} \; \row{f}
& \mapsto & \row{v} = [ u_{\ordIndex^{-1}(k)} \mid 1 \le k\le \rdim\maxDeg_1\cdots\maxDeg_\nvars   ]
  \end{array}
\]
called \emph{expansion}, with inverse $\mathcal{C}_{\ord,\maxDegs}$
called \emph{contraction}.  For a polynomial \(\rel
\in \relSpace_{<\maxDegs}\), $\mathcal{E}_{\ord,\maxDegs}(\rel)$ is
the vector in $\matRing[1][\rdim\maxDeg_1\cdots\maxDeg_\nvars]$ whose
entry at index \(\ordIndex(\vars^\expnts \coordVec{i})\) is the
coefficient of the term involving \(\vars^\expnts \coordVec{i}\) in
\(\rel\).

\begin{example}
  \label{eg:linearization_multiv}
  Consider the case with \(\nvars=2\) variables and \(\rdim=2\), using the
  \(\ordLex\)-term over position order \(\ordTOP{\ordLex}\) on
  \(\bvPolRing^{2}\), with \(Y\ordLex X\). Choosing the degree bounds
  \(\maxDegs = (2,3)\), the monomials
  \[
    \{X^j Y^k \coordVec{i} \mid \tuple{0} \le (j,k) < (2,3), 1 \le i \le 2\}
  \]
  are indexed as follows, according to \cref{dfn:indexing_function}:
  \[
    \begin{array}{cp{1cm}c}
      \text{Monomial}                           &  & \text{Index} \\
      X^j Y^k \coordVec{i}                       &  & \ordIndex[\ordTOP{\ordLex},(2,3)](X^j Y^k \coordVec{i}) \\[0.2cm] \hline \\[-0.3cm]
      \begin{bmatrix} 1    & 0    \end{bmatrix} &  & 1  \\[0.1cm]
      \begin{bmatrix} 0    & 1    \end{bmatrix} &  & 2  \\[0.1cm]
      \begin{bmatrix} Y    & 0    \end{bmatrix} &  & 3  \\[0.1cm]
      \begin{bmatrix} 0    & Y    \end{bmatrix} &  & 4  \\[0.1cm]
      \begin{bmatrix} Y^2  & 0    \end{bmatrix} &  & 5  \\[0.1cm]
      \begin{bmatrix} 0    & Y^2  \end{bmatrix} &  & 6  \\[0.1cm]
      \begin{bmatrix} X    & 0    \end{bmatrix} &  & 7  \\[0.1cm]
      \begin{bmatrix} 0    & X    \end{bmatrix} &  & 8  \\[0.1cm]
      \begin{bmatrix} XY   & 0    \end{bmatrix} &  & 9  \\[0.1cm]
      \begin{bmatrix} 0    & XY   \end{bmatrix} &  & 10  \\[0.1cm]
      \begin{bmatrix} XY^2 & 0    \end{bmatrix} &  & 11 \\[0.1cm]
      \begin{bmatrix} 0    & XY^2 \end{bmatrix} &  & 12
    \end{array}
  \]
  Let \(\row{p}\) be the polynomial in \(\bvPolRing^{2}_{<(2,3)}\) and
  \(\row{v}\) be the vector in \(\matRing[1][12]\) defined by
  \begin{align*}
    \row{p} & = \begin{bmatrix} 46 + 95Y + 75X + 10 XY\; & \;36 + 18Y + 38Y^2 + 77X + 83 XY + 35 XY^2 \end{bmatrix}, \\
    \row{v} & = \begin{bmatrix} 86 & 0 & 32 & 83 & 54 & 26 & 0 & 68 & 86 & 0 & 54 & 22 \end{bmatrix}.
  \end{align*}
  In this case, the expansion of \(\row{p}\) and the contraction of \(\row{v}\)
  are given by
  \begin{align*}
    \expansion{\row{p}} & = \begin{bmatrix} 46 & 36 & 95 & 18 & 0 & 38 & 75 & 77 & 10 & 83 & 0 & 35 \end{bmatrix}, \\
    \contraction{\row{v}} & =  \begin{bmatrix} 86 + 32Y + 54Y^2 + 86XY + 54XY^2\; & \;83Y + 26 Y^2 + 68X + 22XY^2 \end{bmatrix}.
    \qedhere
  \end{align*}
\end{example}

Now, we detail the construction of the multi-Krylov matrix. Let
\(\mulmats = (\mulmat{1},\ldots,\mulmat{\nvars})\) be pairwise
commuting matrices in \(\matRing[\vsdim]\) that define a
$\field[\vars]$-module structure on $\field^{1\times D}$, and let
\(\evMat\) be in \(\matRing[\rdim][\vsdim]\), with rows
$\row{f}_1,\dots,\row{f}_m$. As mentioned in \cref{dfn:syzygy_module}, for
a polynomial \(\rel = [p_1,\ldots,p_\rdim] \in \relSpace\) we write
\[
  \rel \mul \evMat
  = p_1 \mul \row{f}_{1} + \cdots + p_\rdim \mul \row{f}_{\rdim} =
    \row{f}_{1}\, p_1(\mulmats) + \cdots + \row{f}_{\rdim}\, p_\rdim(\mulmats).
\]
As a result, a polynomial \(\rel\) is in \(\modRel\) if its coefficients form a
\(\field\)-linear combination of vectors of the form \(\row{f}_{i}
\mulmats^\expnts\) which is zero. If furthermore \(\rel\) is nonzero and has
its degrees in each variable bounded by \((\maxDeg_1,\ldots,\maxDeg_\nvars)\),
then it corresponds to a nontrivial \(\field\)-linear relation between the row
vectors
\[
  \{ \row{f}_{i}  \mulmats^\expnts \mid \tuple{0} \le \expnts < \maxDegs, 1 \le i \le \rdim\}.
\]
This leads us to consider the matrices formed by these vectors, ordered
according to \(\ordIndex\).

\begin{definition}
  \label{dfn:multi_krylov}
  Let \(\ord\) be a monomial order on \(\relSpace\), let \(\mulmats =
  (\mulmat{1},\ldots,\mulmat{\nvars}) \in \matRing[\vsdim]\) be pairwise
  commuting matrices, let \(\evMat \in \matRing[\rdim][\vsdim]\) whose rows are
  \(\row{f}_1,\ldots,\row{f}_\rdim\), and let
  \(\maxDegs=(\maxDeg_1,\ldots,\maxDeg_\nvars) \in \NNp^\nvars\). The
  \emph{\((\ord,\maxDegs)\)-multi-Krylov matrix for \((\mulmats,\evMat)\)},
  denoted by \(\krylov{\ord,\maxDegs}\), is the matrix in
  \(\matRing[\rdim\maxDeg_1\cdots\maxDeg_\nvars][\vsdim]\) whose row at index
  \(\ordIndex(\vars^\expnts \coordVec{i})\) is \(\vars^\expnts \coordVec{i} \cdot
  \evMat= \row{f}_{i} \mulmats^\expnts\), for \(\tuple{0} \le \expnts <
  \maxDegs\) and \(1 \le i \le \rdim\).
\end{definition}

\begin{example}
  \label{eg:krylovmat}
  Following on from \cref{eg:linearization_multiv}, we consider the vector
  space dimension \(\vsdim = 3\) and matrices \(\evMat\) in \(\matRing[2][3]\)
  and \(\mulmats = (\mulmat{X},\mulmat{Y})\) in \(\matRing[3]\) such that
  \(\mulmat{X} \mulmat{Y} = \mulmat{Y}\mulmat{X}\).  Then from the indexing
  function \(\ordIndex[\ordTOP{\ordLex},(2,3)]\) described above we obtain
  \[
    \krylov{\ordTOP{\ordLex},(2,3)} =
    \begin{bmatrix}
      \evMat \\
      \evMat \mulmat{Y} \\
      \evMat \mulmat{Y}^2 \\
      \evMat \mulmat{X} \\
      \evMat \mulmat{X}\mulmat{Y} \\
      \evMat \mulmat{X}\mulmat{Y}^2
    \end{bmatrix}
    \in \matRing[12][3]. \qedhere
  \]
\end{example}

By construction, we have the following result, which relates the left nullspace
of the multi-Krylov matrix with the set of bounded-degree syzygies.

\begin{lemma}
  \label{lem:krylov_nullspace}
  If \(\row{v} \in \matRing[1][\rdim\maxDeg_1\cdots\maxDeg_\nvars]\) is in the
  left nullspace of \(\krylov{\ord,\maxDegs}\), then its contraction
  \(\contraction{\row{v}} \in \relSpace_{<\maxDegs}\) is in \(\modRel\).
  Conversely, if \(\rel \in \relSpace_{<\maxDegs}\) is in \(\modRel\), then
  its expansion \(\expansion{\rel} \in
  \matRing[1][\rdim\maxDeg_1\cdots\maxDeg_\nvars]\) is in the left nullspace
  of \(\krylov{\ord,\maxDegs}\).
\end{lemma}

Our first step towards finding a \(\ord\)-Gr\"obner basis \(\gb\) of
\(\modRel\) consists in computing the \(\ord\)-monomial basis $\monbas$
of the quotient \(\relSpace/\modRel\); we are going to prove that it
corresponds to the row rank profile of the multi-Krylov matrix. From
the above discussion, we know that considering this matrix only gives
us access to syzygies which satisfy degree constraints.  In what
follows, we will choose $\maxDegs=(\maxDeg_1,\dots,\maxDeg_r)$, with
$\maxDeg_i \ge D$ for all $i$. In particular, for this choice, all
elements in the monomial basis of \(\relSpace/\modRel\) are in
$\relSpace_{<\maxDegs}$.

We recall that, for a matrix \(\mat{A}\) in \(\matRing[\mu][\nu]\), the
\emph{row rank profile} of \(\mat{A}\) is the lexicographically smallest
subtuple \((\rkprof_1,\ldots,\rkprof_\dimvs)\) of \((1,\ldots,\mu)\) such that
\(\dimvs\) is the rank of \(\mat{A}\) and the rows
\((\rkprof_1,\ldots,\rkprof_\dimvs)\) of \(\mat{A}\) are linearly independent.

\begin{theorem}
  \label{thm:rrp_monbas}
  Let \(\ord\) be a monomial order on \(\relSpace\), let \(\mulmats
  = (\mulmat{1},\ldots,\mulmat{\nvars})\) be pairwise commuting
  matrices in \(\matRing[\vsdim]\), let \(\evMat \in
  \matRing[\rdim][\vsdim]\), and let
  $\maxDegs=(\maxDeg_1,\dots,\maxDeg_r)$, with $\maxDeg_i \ge D$ for
  all $i$. Let further \((\rkprof_1,\ldots,\rkprof_\dimvs) \in
  \NNp^\dimvs\) be the row rank profile of
  \(\krylov{\ord,\maxDegs}\). Then the \(\ord\)-monomial basis
  $\monbas$ of \(\relSpace/\modRel\) is equal to
  \(\{\ordIndex^{-1}(\rkprof_1), \ldots,
  \ordIndex^{-1}(\rkprof_\dimvs)\}\).
\end{theorem}
\begin{proof}
  Write \(\rkprof_j = \ordIndex(\vars^{\expnts_j}\coordVec{i_j})\), so that
  \(\{\ordIndex^{-1}(\rkprof_1), \ldots, \ordIndex^{-1}(\rkprof_\dimvs)\} =
  \{\vars^{\expnts_1} \coordVec{i_1}, \ldots, \vars^{\expnts_\dimvs}
  \coordVec{i_\dimvs}\}\). We want to prove that \(\lt{\modRel}\) is the set
  of monomials not in \(\{\vars^{\expnts_1} \coordVec{i_1}, \ldots,
  \vars^{\expnts_\dimvs} \coordVec{i_\dimvs}\}\).

  First, consider any monomial \(\vars^\expnts \coordVec{i}\) for \(1\le i \le
  \rdim\) and \(\expnts \in \NN^\nvars\) such that \(\expnts \not< \maxDegs\).
  Such a monomial cannot be in \(\{\vars^{\expnts_1} \coordVec{i_1}, \ldots,
  \vars^{\expnts_\dimvs} \coordVec{i_\dimvs}\}\) since by construction of
  \(\krylov{\ord,\maxDegs}\) we have \(\expnts_j < \maxDegs\) for all \(j\). On
  the other hand, \(\vars^\expnts \coordVec{i}\) cannot be in the
  \(\ord\)-monomial basis $\monbas$ either. Indeed, writing
  $\vars^\expnts=X_1^{e_1} \cdots X_r^{e_r}$, $\expnts \not< \maxDegs$ means
  that $e_k \ge \maxDeg_k \ge D$ for some index $k$; if \(\vars^\expnts
  \coordVec{i}\) is in $\monbas$ then $X_k^{e_k} \coordVec{i}$ is also in
  $\monbas$ and thus $\coordVec{i},X_k\coordVec{i}, \ldots, X_k^\vsdim
  \coordVec{i}$ are in $\monbas$, which is a contradiction since linear
  independence would imply that the minimal polynomial of $\mulmat{k}$ has
  degree greater than $\vsdim$. Hence \(\vars^\expnts \coordVec{i} \in
  \ltMod{\modRel}\).

  Now, let \(\vars^\expnts \coordVec{i} \in \ltMod{\modRel}\) be such that
  \(\expnts < \maxDegs\). Then there is a polynomial \(\rel\) in \(\modRel\)
  such that \(\lt{\rel} = \vars^\expnts \coordVec{i}\), and
  \cref{lem:krylov_nullspace} implies that \(\expansion{\rel}\) is in the left
  nullspace of \(\krylov{\ord,\maxDegs}\). Since by construction the rightmost
  nonzero entry of \(\expansion{\rel}\) is \(1\) at index
  \(\ordIndex(\vars^\expnts \coordVec{i})\), the vector \(\expansion{\rel}\)
  expresses the row of \(\krylov{\ord,\maxDegs}\) with index
  \(\ordIndex(\vars^\expnts \coordVec{i})\) as a \(\field\)-linear combination
  of the rows with smaller indices. By definition of the row rank profile, this
  implies \(\ordIndex(\vars^\expnts \coordVec{i}) \not\in
  \{\rkprof_1,\ldots,\rkprof_\dimvs\}\), and therefore \(\vars^\expnts
  \coordVec{i} \not\in \{\vars^{\expnts_1} \coordVec{i_1}, \ldots,
  \vars^{\expnts_\dimvs} \coordVec{i_\dimvs}\}\).

  Conversely, let \(\vars^\expnts \coordVec{i} \not\in \{\vars^{\expnts_1}
  \coordVec{i_1}, \ldots, \vars^{\expnts_\dimvs} \coordVec{i_\dimvs}\}\) be a
  monomial such that \(\expnts < \maxDegs\). Then \(\ordIndex(\vars^\expnts
  \coordVec{i}) \not\in \{\rkprof_1,\ldots,\rkprof_\dimvs\}\). Thus, by
  definition of the row rank profile, there is a nonzero vector \(\row{v} \in
  \matRing[1][\rdim D^\nvars]\) such that \(\row{v}\) is in the left nullspace
  of \(\krylov{\ord,\maxDegs}\) and the rightmost nonzero entry of \(\row{v}\)
  is \(1\) at index \(\ordIndex(\vars^\expnts \coordVec{i})\).  Then
  \(\lt{\contraction{\row{v}}}=\vars^\expnts \coordVec{i}\), and according to
  \cref{lem:krylov_nullspace}, \(\contraction{\row{v}}\) is in \(\modRel\),
  hence \(\lt{\contraction{\row{v}}} \in \ltMod{\modRel}\).
\end{proof}

In particular, we see that the dimension $\dimvs$ of \(\relSpace/\modRel\) as
a \(\field\)-vector space is equal to the rank of the multi-Krylov matrix
\(\krylov{\ord,\maxDegs}\). In particular this implies $\dimvs \le \vsdim$,
since \(\krylov{\ord,\maxDegs}\) has \(\vsdim\) columns.


\subsection{Computing the monomial basis}
\label{sec:syzygy:monbas}

We now show how to exploit the structure of the multi-Krylov matrix so as to
efficiently compute its row rank profile, yielding the \(\ord\)-monomial basis
$\monbas$ of \(\relSpace/\modRel\).

For $\maxDegs$ as in~\cref{thm:rrp_monbas}, the dense representation of
\(\krylov{\ord,\maxDegs}\) uses at least \(\rdim \vsdim^{\nvars+1}\) field
elements, which is well beyond our target cost \(\softO{\rdim
\vsdim^{\expmatmul-1} + \nvars \vsdim^\expmatmul}\). On the other hand, this
matrix is succinctly described by the data \((\ord,\mulmats,\evMat)\), which
requires \(\bigO{\rdim \vsdim + \nvars \vsdim^2}\) field elements. Like
previous related algorithms \cite{MolBuc82,FaGiLaMo93,MaMoMo93}, we will never
compute the full dense representation of this matrix, but rather always store
and use a minimal amount of data that allows the algorithm to progress. The
main property behind this is that once some monomial is found not to be in the
sought monomial basis, then all multiples of that monomial can be discarded
from the rest of the computation.

Our algorithm for computing the monomial basis $\monbas=(\basVec{1},
\ldots,\basVec{\dimvs})$ also returns the matrix $\mat{B} \in
\matRing[\dimvs][\vsdim]$ whose rows are $\basVec{1} \mul
\evMat,\dots,\basVec{\dimvs} \mul \evMat$; it will be needed later on.
The algorithm exploits the structure of this matrix, and uses fast
arithmetic for matrices over \(\field\) through
\begin{itemize}
  \item a procedure \algoname{RowRankProfile} which computes the row rank
    profile of any matrix in \(\matRing[\mu][\nu]\) of rank \(\rho\) in
    \(\bigO{\rho^{\expmatmul-2} \mu \nu}\) operations in \(\field\), as
    described in \cite[Sec.\,2.2]{Storjohann00};
  \item matrix multiplication which is incorporated by following a strategy in
    the style of Keller-Gehrig \cite{KelGeh85}.
\end{itemize}

In short, the latter strategy can be thought of as precomputing powers of the
form \(\mulmat{j}^{2^e}\) of the multiplication matrices, which then allows us
to group many vector-matrix products into a small number of matrix-matrix
products. To achieve this we work iteratively on the variables, thus first
focusing on all operations involving \(\mulmat{1}\), then on those involving
\(\mulmat{2}\), etc. The order of the rows specified by \(\ord\) is not
respected in this process, since at a fixed stage of the algorithm we will only
have considered a submatrix of the multi-Krylov matrix which does not involve
the last variables. To fix this we constantly re-order, according to \(\ord\),
the rows that have been processed and the ones that we introduce.

\begin{algobox}
  \algoInfo
  {MonomialBasis}
  {algo:monomial_basis}

  \dataInfos{Input}{
    \item monomial order \(\ord\) on \(\ring^{\rdim}=\mvPolRing{\nvars}^{\rdim}\),
    \item pairwise commuting matrices \(\mulmats = (\mulmat{1},\ldots,\mulmat{\nvars})\) in \(\matRing[\vsdim]\), 
    \item matrix \(\evMat \in \evSpace{\vsdim}\). }

  \dataInfos{Output}{
    \item the \(\ord\)-monomial basis $\monbas=(\basVec{1},\dots,\basVec{\dimvs})$ of \(\relSpace/\modRel\),
    \item the matrix $\mat{B} \in \matRing[\dimvs][\vsdim]$ whose rows are $\basVec{1} \mul
    \evMat,\dots,\basVec{\dimvs} \mul \evMat$. }

  \algoSteps{
  \item $\maxDeg \assign 2^{\lceil \log_2(\vsdim) \rceil + 1}$;~ $\maxDegs \assign (\maxDeg,\dots,\maxDeg) \in \NNp^\nvars$
  \item \(\ordIndex \assign\) the indexing function in \cref{dfn:indexing_function}
  \item \(\pi \assign\) the permutation matrix in \(\{0,1\}^{\rdim\times\rdim}\)
    such that the tuple \(\tuple{t} = \pi
    \trsp{[\ordIndex(\coordVec{1}),\ldots,\ordIndex(\coordVec{\rdim})]}\) is
    increasing
  \item \(\basMat \assign \pi \evMat\)
  \item \(\hat \delta,(i_1,\ldots,i_{\hat \delta}) \assign \algoname{RowRankProfile}(\basMat)\)
  \item \algoword{For} \(k\) \algoword{from} \(1\) \algoword{to} \(\nvars\) \eolcomment{iterate over the variables}
    \begin{algosteps}[{\bf a.}]
      \item \(\mat{P} \assign \mulmat{k}\);~ $\expnt \assign 0$
      \item \algoword{Do}
        \begin{algosteps}[(i)]
          \item \(\delta \assign \hat\delta\)
          \item \((\rkprof_1,\ldots,\rkprof_{\delta}) \assign\) subtuple of \(\tuple{t}\) with  entries \(i_1,\ldots,i_{\delta}\)
          \item \(\basMat \assign\) the submatrix of \(\basMat\) with rows \(i_1,\ldots,i_{\delta}\)
          \item \(\hat{\rkprof}_j \assign \ordIndex(\var_k^{2^\expnt}\ordIndex^{-1}(\rkprof_j))\) for \(1\le j\le \delta\)
          \item \(\pi \assign\) the permutation matrix in \(\{0,1\}^{2\delta\times2\delta}\) such that the tuple
            \(\tuple{t} = \pi \trsp{[\rkprof_1,\ldots,\rkprof_\delta,\hat\rkprof_1,\ldots,\hat\rkprof_\delta]}\) is increasing
          \item \(\basMat \assign \pi \begin{bmatrix} \basMat \\ \basMat \mat{P} \end{bmatrix}\)
          \item \(\hat\delta, (i_1,\ldots,i_{\hat\delta}) \assign \algoname{RowRankProfile}(\basMat)\)
          \item \(\mat{P} \assign \mat{P}^2\); $\expnt \assign \expnt +1$
        \end{algosteps}
      \item[] \algoword{Until} $\hat\delta=\delta$ and \((\rkprof_1,\ldots,\rkprof_{\delta}) =\) subtuple of \(\tuple{t}\) with entries \(i_1,\ldots,i_{\delta}\)
    \end{algosteps}
    \item \algoword{Return} \(\monbas=(\ordIndex^{-1}(\rkprof_1), \ldots, \ordIndex^{-1}(\rkprof_\delta))\)
    and the submatrix of \(\basMat\) with rows \(i_1,\ldots,i_{\delta}\) }
\end{algobox}

\begin{proposition}
  \label{prop:algo:monomial_basis}
  \cref{algo:monomial_basis} returns the \(\ord\)-monomial basis
  $\monbas=(\basVec{1},\dots,\basVec{\dimvs})$ of \(\relSpace/\modRel\) and
  the matrix $\mat{B} \in \matRing[\dimvs][\vsdim]$ whose rows are
  $\basVec{1} \mul \evMat,\dots,\basVec{\dimvs} \mul \evMat$. It uses
  \[
    \bigO{\rdim \vsdim^{\expmatmul-1} +  \vsdim^\expmatmul (\nvars + \log(d_1 \cdots d_\nvars))} 
    \subset \bigO{\rdim \vsdim^{\expmatmul-1} + \nvars \vsdim^\expmatmul \log(\vsdim)}
  \]
  operations in \(\field\), where $d_k \in \{1,\ldots,\vsdim\}$ is the degree
  of the minimal polynomial of $\mulmat{k}$, for $1 \le k \le \nvars$.
\end{proposition}
\begin{proof}\renewcommand{\qedsymbol}{}
  Let $\maxDeg=2^{\lceil \log_2(\vsdim) \rceil + 1}$ as in the algorithm; for
  \(1\le k\le \nvars\) and $0 \le \expnt \le \log_2(\maxDeg)$, let us consider
  the set of monomials
  \[
    \expSet_{k,\expnt} = \{\vars^{\expnts} \coordVec{i} \mid 1\le i \le
    \rdim, \; 0\le \expnts < (\maxDeg,\ldots,\maxDeg,2^e,1,\ldots,1)\},
  \]
  where \(2^e\) is the \(k\)th entry of the tuple. Then we denote by
  \(\mat{C}_{k,\expnt} \in \matRing[(\rdim \maxDeg^{k-1}2^e)][\vsdim]\) the
  submatrix of \(\krylov{\ord,\maxDegs}\) formed by its rows in
  \(\ordIndex(\expSet_{k,\expnt})\). For $0 \le \expnt \le \expnt' \le
  \log_2(\maxDeg)$, \(\mat{C}_{k,\expnt}\) is a submatrix of
  \(\mat{C}_{k,\expnt'}\), but not necessarily a top submatrix of it. Remark
  also that for $k < r$, $\mat{C}_{k,\log_2(\maxDeg)} = \mat{C}_{k+1,0}$.

  The correctness of the algorithm is proved by an induction that involves $k$ and
  $e$. For $1 \le k \le r$, denote by $\ell_k \ge 0$ the last value of the index
  $\expnt$ for which we enter the body of the \algoword{Do-Until} loop
  (Step~\textbf{6.b}). Then, for $1\le k\le r$ and $0 \le e \le \ell_k$, define
  the following assertions, that we consider at the beginning iteration
  \((k,\expnt)\) of the \algoword{Do-Until} loop:
  \begin{itemize}
    \item [$A_1:$] the rows of indices $i_1,\dots,i_{\hat
      \delta}$ in \(\basMat\) are the rows defining the row
      rank profile of \(\mat{C}_{k,e}\);
    \item [$A_2:$] the entries of indices $i_1,\dots,i_{\hat \delta}$ in
      $\tuple t$ are the indices of these rows in
      \(\krylov{\ord,\maxDegs}\).
  \end{itemize}
  We will prove by induction that these properties hold for all values of $k$
  and $e$ considered above.  For $k=1$ and $e=0$, $\mat{C}_{1,0}$ is the
  submatrix of \(\krylov{\ord,\maxDegs}\) with rows in
  \(\{\ordIndex(\coordVec{1}), \ldots, \ordIndex(\coordVec{\rdim})\}\);
  therefore, by choice of the permutation \(\pi\) at Step~\textbf{3}, we have
  \(\mat{C}_{1,0} = \pi\evMat=\mat{B}\) at Step~\textbf{4}. Thus, upon entering
  the \algoword{For} loop for the first time, with $k=1$ and $e=0$, we see that
  $A_1$ and $A_2$ hold.

  Then let \(k\) be in \(\{1,\ldots,\nvars\}\) and $\expnt$ in
  $\{0,\dots,\ell_k\}$; we assume that $A_1$ and $A_2$ hold at indices $k$ and
  $e$. Let us denote \(\rkprofs = \{\rkprof_1,\ldots,\rkprof_\delta\}\) as
  defined in Step~\textbf{6.b.}(ii), and let \(\hat\rkprofs =
  \{\hat\rkprof_1,\ldots,\hat\rkprof_\delta\}\), where \(\hat{\rkprof}_j =
  \ordIndex(\var_k^{2^\expnt}\ordIndex^{-1}(\rkprof_j))\) for \(1\le j\le
  \delta\) are the indices computed at Step~\textbf{6.b.}(iv). Let also
  \((\gamma_1,\ldots,\gamma_\nu)\) be the indices of the rows of
  \(\krylov{\ord,\maxDegs}\) corresponding to the row rank profile of its
  submatrix \(\mat{C}_{k,\expnt+1}\). To complete this proof, we will use three
  intermediate lemmas; first, we claim that the following holds:
\end{proof}
  \begin{lemma}
    \label{lem:gamma_subsequence}
    \((\gamma_1,\ldots,\gamma_\nu)\) is a subsequence of the tuple
    \(\tuple{t}\).
  \end{lemma}
  \begin{proof}[Proof of \cref{lem:gamma_subsequence}]
    Let $j$ be in $\{1,\dots,\nu\}$ and let us prove that \(\gamma_j\) is in
    \(\rkprofs\cup\hat{\rkprofs}\). By assumption, the row of index $\gamma_j$
    in \(\krylov{\ord,\maxDegs}\) is not a linear combination of the rows of
    smaller indices in the submatrix \(\mat{C}_{k,\expnt+1}\) of
    \(\krylov{\ord,\maxDegs}\).

    Suppose first that \(\ordIndex^{-1}(\gamma_j)\) is in
    \(\expSet_{k,\expnt}\), so that $\gamma_j$ actually corresponds to a row in
    \(\mat{C}_{k,\expnt}\).  Since \(\mat{C}_{k,\expnt}\) is a submatrix of
    \(\mat{C}_{k,\expnt+1}\), the remark above implies that the row of index
    $\gamma_j$ in \(\krylov{\ord,\maxDegs}\) is not a linear combination of
    the rows of smaller indices in \(\mat{C}_{k,\expnt}\). This means that
    this row belongs to the row rank profile of \(\mat{C}_{k,\expnt}\), and so
    (by $A_2$) $\gamma_j$ is in \(\rkprofs\).

    Now, we assume that \(\ordIndex^{-1}(\gamma_j) \in \expSet_{k,\expnt+1} -
    \expSet_{k,\expnt}\), and we prove that \(\gamma_j \in \hat\rkprofs\), or
    in other words, that \(\ordIndex^{-1}(\gamma_j) \in \{X_k^{2^e}
    \ordIndex^{-1}(\rkprof_j) \mid 1\le j\le \delta\}\). Since
    \(\ordIndex^{-1}(\gamma_j) \in \expSet_{k,\expnt+1} - \expSet_{k,\expnt}\),
    we can write \(\ordIndex^{-1}(\gamma_j) = X_k^{2^e} \vars^{\bm{f}}
    \coordVec{i}\), with \(\vars^{\bm{f}} \coordVec{i}\) in
    \(\expSet_{k,\expnt}\). Suppose that \(\vars^{\bm{f}} \coordVec{i}\) is not
    in $\ordIndex^{-1}(\rkprofs)$, so that by $A_2$, the row of index
    $\ordIndex(\vars^{\bm{f}} \coordVec{i})$ in $\krylov{\ord,\maxDegs}$ is a
    linear combination of the previous rows in $\mat{C}_{k,\expnt}$.
    Right-multiply all these rows by $\mulmat{k}^{2^e}$; this shows that the
    row indexed by $\gamma_j$ in \(\krylov{\ord,\maxDegs}\) is a linear
    combination of rows of smaller indices in its submatrix
    \(\mat{C}_{k,\expnt+1}\), a contradiction.  Our claim is proved.
  \end{proof}

  Now, by $A_1$ and $A_2$, after the update at Steps~\textbf{6.b.}(vi), the
  rows of $\mat{B}$ are precisely the rows of \(\krylov{\ord,\maxDegs}\) of
  indices in $\rkprofs \cup \hat\rkprofs$ sorted in increasing order.  Thus,
  after the row rank profile computation at Step~\textbf{6.b.}(vii), the rows
  in $\mat{B}$ of indices \(i_1,\ldots,i_{\hat\delta}\) are the rows of
  $\mat{C}_{k,e+1}$ corresponding to its row rank profile, and the subtuple of
  \(\tuple{t}\) with entries \(i_1,\ldots,i_{\hat\delta}\) is precisely
  $(\gamma_1,\dots,\gamma_\nu)$. 

  If $e < \ell_k$, this implies that $A_1$ and $A_2$ still hold at step
  $(k,e+1)$. Suppose next that instead, $e=\ell_k$. We claim the following.

  \begin{lemma}\label{lemma:logdk}
    We have $\ell_k \le \lceil \log_2(d_k) \rceil$; equivalently, if we exit
    the \algoword{Do-Until} loop at the end of iteration $e$, then $e \le
    \lceil \log_2(d_k) \rceil$.
  \end{lemma}
  \begin{proof}[Proof of \cref{lemma:logdk}]
    First, we observe that the indices of the rows in
    \(\krylov{\ord,\maxDegs}\) corresponding to the row rank profile of
    $\mat{C}_{k,e-1}$, and of those corresponding to the row rank profile of
    $\mat{C}_{k,e}$, are different. Indeed, $A_1$ and $A_2$ show that the
    former correspond to indices $(\rho_1,\dots,\rho_\delta)$ obtained at
    Step~\textbf{6.b.}(ii) at iteration $e-1$, the latter to the same indices
    at iteration $e$, and the fact that we did not exit the loop at step $e-1$
    implies that they are different.

    Suppose then, by means of contradiction, that $e \ge \lceil \log_2(d_k)
    \rceil+1$, so that $e \ge \log_2(d_k) + 1$. Consider a row $\row{r}$ in
    $\mat{C}_{k,e}$ that is not in $\mat{C}_{k,e-1}$; then, $\row{r} = \row{s}
    \mulmat{k}^{2^{e-1}}$ for some row $\row{s}$ in $\mat{C}_{k,e-1}$.  By
    assumption, $2^{e-1}$ is at least equal to the degree $d_k$ of the minimal
    polynomial of $\mulmat{k}$. In particular, $\mulmat{k}^{2^{e-1}}$ is a
    linear combination of powers of $\mulmat{k}$ of exponent less than
    $2^{e-1}$. Now, all rows $\row{s} \mulmat{k}^i$, for $i < 2^{e-1}$, are in
    $\mat{C}_{k,e}$, and have lower indices than $\row{r}$. This implies that
    $\row{r}$ is not in the row rank profile of $\mat{C}_{k,e}$, and thus
    $\mat{C}_{k,e-1}$ and $\mat{C}_{k,e}$ have the same row rank profile. This
    contradicts the property in the previous paragraph, hence $e \le \lceil
    \log_2(d_k) \rceil$.
  \end{proof}

  Using this property, we prove that $A_1$ and $A_2$ now hold for indices $k+1$
  and $e=0$ (this will be enough to conclude our induction proof).  
  \begin{lemma}
    \label{lem:rows_are_rows}
    If we exit the \algoword{Do-Until} loop after step $e$ (equivalently, if
    $e=\ell_k$), then the rows in $\mat{B}$ of indices
    \(i_1,\ldots,i_{\hat\delta}\) are the rows of $\mat{C}_{k,\log_2(\maxDeg)}$
    corresponding to its row rank profile, and the subtuple of \(\tuple{t}\)
    with entries \(i_1,\ldots,i_{\hat\delta}\) is the indices of these rows in
    \(\krylov{\ord,\maxDegs}\).
  \end{lemma}
  \begin{proof}[Proof of \cref{lem:rows_are_rows}]
    Our assumption means that \((\gamma_1,\ldots,\gamma_\nu) =
    (\rho_1,\dots,\rho_\delta)\); this is equivalent to saying that the indices
    in \(\krylov{\ord,\maxDegs}\) of the row rank profiles of its submatrices
    \(\mat{C}_{k,\expnt}\) and \(\mat{C}_{k,\expnt+1}\) are the same. In
    particular, any row in \(\mat{C}_{k,\expnt+1}\) is a linear combination of
    rows of lower indices in \(\mat{C}_{k,\expnt}\).

    We will prove the following below: \emph{any row in
    \(\mat{C}_{k,\log_2(\maxDeg)}\) is a linear combination of rows of lower
  indices in \(\mat{C}_{k,\expnt}\)}. In particular, this implies that the
  indices in \(\krylov{\ord,\maxDegs}\) of the row rank profiles of its
  submatrices \(\mat{C}_{k,e+1}\) and \(\mat{C}_{k,\log_2(\maxDeg)}\) are all
  the same.  Since we saw that after Step~\textbf{6.b.}(vii), the rows in
  $\mat{B}$ of indices \(i_1,\ldots,i_{\hat\delta}\) are the rows of
  $\mat{C}_{k,e+1}$ corresponding to its row rank profile, and the subtuple of
  \(\tuple{t}\) with entries \(i_1,\ldots,i_{\hat\delta}\) are the indices of
  these rows in \(\krylov{\ord,\maxDegs}\), this is enough to prove the lemma.

    We prove our claim by induction on the rows of
    $\mat{C}_{k,\log_2(\maxDeg)}$. Let thus $\row{r}$ be a row in
    $\mat{C}_{k,\log_2(\maxDeg)}$, and assume the claim holds for all previous
    rows.

    If $\row{r}$ is from its submatrix $\mat{C}_{k,e}$, we are done. Else,
    since $e+1 \le \log_2(\maxDeg)$ holds (from the previous lemma), $\row{r}$
    can be written as $\row{r}=\row{s} \mat{M}_k^c$, for some $c \ge 0$, where
    $\row{s}$ is a row in $\mat{C}_{k,e+1}$. We know that $\row{s}$ is a linear
    combination of rows $\row{s}_1,\dots,\row{s}_\iota$ of lower indices in
    $\mat{C}_{k,e}$, so that $\row{r}$ is a linear combination of $\row{s}_1
    \bm{M}_k^c,\dots,\row{s}_\iota \bm{M}_k^c$. All these rows are in
    $\mat{C}_{k,\log_2(\maxDeg)}$ and have lower indices than $\row{r}$. By our
    induction assumption, they are linear combinations of rows of lower indices
    in \(\mat{C}_{k,\expnt}\), and thus so is $\row{r}$.
  \end{proof}

  \begin{proof}[(Continuing proof of \cref{prop:algo:monomial_basis}.)]
  If \(k<\nvars\) then we can turn to the next variable \(\var_{k+1}\), since
  the former lemma, together with the equality \(\mat{C}_{k,\log_2(\maxDeg)} =
  \mat{C}_{k+1,0}\), shows that $A_1$ and $A_2$ hold for indices $(k+1,0)$.
  For $k=r$, since \(\mat{C}_{r,\log_2(\maxDeg)}=\krylov{\ord,\maxDegs}\), the
  lemma shows that the output of the algorithm is indeed the row rank profile
  of \(\krylov{\ord,\maxDegs}\) and the submatrix of \(\krylov{\ord,\maxDegs}\)
  formed by the corresponding rows.  According to \cref{thm:rrp_monbas}, the
  \(\ord\)-monomial basis can directly be deduced from the rank profile of
  \(\krylov{\ord,\maxDegs}\). This concludes the proof of correctness.

  Concerning the cost bound, according to \cite[Thm.\,2.10]{Storjohann00}, the
  row rank profile computation at Step~\textbf{5} can be computed in
  \(\bigO{\matRank^{\expmatmul-2} \rdim \vsdim}\) operations, where \(\matRank
  \le D\) is the rank of \(\evMat\). In particular, this is \(\bigO{\rdim
  \vsdim^{\expmatmul-1}}\).

  Let us now focus on the iteration \((k,\expnt)\) and we show that it uses
  \(\bigO{\vsdim^\expmatmul}\) operations. First, note that $\delta \le D$
  holds throughout (since $\delta$ is the rank of a matrix with $D$ columns).
  Since upon entering Step~\textbf{6.b.}(vi), \(\basMat\) has \(\delta \le
  \vsdim\) rows and \(\vsdim\) columns, its update and the row rank profile of
  the latter can be computed in \(\bigO{\vsdim^\expmatmul}\) operations.
  Finally, squaring the \(\vsdim\times\vsdim\) matrix \(\mat{P}\) at
  Step~\textbf{6.b.}(viii) is also done in \(\bigO{\vsdim^\expmatmul}\)
  operations.

  To conclude the proof of the cost bound, we recall from \cref{lemma:logdk}
  that in iteration \(k\) of the \algoword{For} loop, we pass $\ell_k+1 \le
  \lceil \log_2(d_k) \rceil + 1$ times through the body of the
  \algoword{Do-Until} loop.
\end{proof}

\begin{remark}
  \label{rmk:krylov_lextop}
  We may slightly refine the analysis for some particular monomial orders.
  Indeed, the order in which the \algoword{For} and \algoword{Do-Until} loops
  introduce the new monomials to be processed corresponds to the
  \(\ordLex\)-term over position order \(\ordTOP{\ordLex}\) over \(\relSpace\),
  with \(X_1 \ordLex \cdots \ordLex X_\nvars\). As a result, the behaviour and
  the cost bound of the algorithm can be described with more precision if the
  input monomial order is \(\ord\; = \ordTOP{\ordLex}\).

  In this case, we are processing the rows of \(\krylov{\ord,\maxDegs}\) in the
  order they are in the matrix. In particular, the permutation \(\pi\) at
  Steps~\textbf{3} and~\textbf{6.b.}(v) is always the identity matrix, and the
  tuple \((\rkprof_1,\ldots,\rkprof_\delta)\) inside the loops consists of the
  first \(\delta\) entries of the actual row rank profile of
  \(\krylov{\ord,\maxDegs}\).

  Furthermore, the fact that we are processing the rows in their actual order
  has a small impact on the cost bound, as follows. Let us denote by \(\gb\)
  the reduced \(\ord\)-Gr\"obner basis of \(\modRel\), and let
  \(\boldsymbol{\hat\maxDeg} = (\hat\maxDeg_1,\ldots,\hat\maxDeg_\nvars)\) be
  the tuple of maximum degrees in \(\gb\), that is, \(\hat\maxDeg_k =
  \max_{\rel \in \gb} \deg_{X_k}(\rel)\) for \(1 \le k \le \nvars\). 

  Then, at iteration \(k\) of the \algoword{For} loop, the \algoword{Do-Until}
  loop does at most \(\lceil \log_2(\hat\maxDeg_k+1) \rceil + 1\) iterations;
  indeed, when reaching the iteration that introduces powers of the variable
  \(X_k\) all greater than \(\hat\maxDeg_k\), the partial row rank profile
  \((\rkprof_1,\ldots,\rkprof_\delta)\) is not modified anymore, and
  the \algoword{Do-Until} loop exits. Now, considering the total number of
  iterations, we claim that
  \begin{equation}
    \label{eqn:upper_bound_sumlog}
    \sum_{1\le k\le \nvars} \log(\hat\maxDeg_k+1)
    \;\le\;
    \nvars \log\!\left(\frac{\hat\maxDeg_1+\cdots+\hat\maxDeg_\nvars}{\nvars}+ 1\right)
    \;\le\;
    \nvars \log\!\left(\frac{\vsdim}{\nvars}+2\right).
  \end{equation}
  As a result, when the input order is \(\ord\; = \ordTOP{\ordLex}\),
  \cref{algo:monomial_basis} uses
  \[
    \bigO{\rdim \vsdim^{\expmatmul-1} + \nvars \vsdim^\expmatmul
    \log\!\left(\frac{\vsdim}{\nvars}+2\right)}
  \]
  operations in \(\field\).

  We now prove our claim. The first inequality in \cref{eqn:upper_bound_sumlog}
  is a direct application of the arithmetic mean-geometric mean inequality.
  The second inequality follows from the bound
  \(\hat\maxDeg_1+\cdots+\hat\maxDeg_\nvars \le \vsdim + \nvars - 1\), which
  holds since the \(\ord\)-monomial basis, whose cardinality \(\dimvs\) is at
  most $\vsdim$, contains the
  \(1+\hat\maxDeg_1+\cdots+\hat\maxDeg_\nvars-\nvars\) distinct monomials
  specified hereafter. By definition of \(\boldsymbol{\hat\maxDeg}\), for each
  \(k \in \{1,\ldots,\nvars\}\) such that \(\hat\maxDeg_k > 0\), there is a
  monomial appearing in some element of \(\gb\) which is a multiple of
  \(X_k^{\hat\maxDeg_k} \coordVec{i_k}\) for some \(1 \le i_k \le \rdim\).
  Thus \(X_k^{\hat\maxDeg_k} \coordVec{i_k}\) is either in the
  \(\ord\)-monomial basis or in its border, which implies that the
  \(\ord\)-monomial basis contains \(\{\coordVec{i_k}, X_k \coordVec{i_k},
  \ldots, X_k^{\hat\maxDeg_k-1} \coordVec{i_k}\}\). Considering the union of
  all such sets for \(1 \le k \le \nvars\) (with the empty set if
  \(\hat\maxDeg_k = 0\)) yields at least
  \(1+(\hat\maxDeg_1-1)+\cdots+(\hat\maxDeg_\nvars-1)\) distinct monomials,
  since each intersection of a pair of such sets has at most one element (the
  coordinate vector).
\end{remark}

\subsection{Fast computation of the basis of syzygies}
\label{sec:syzygy:relbas}

Next, we present our algorithm to compute the reduced Gr\"obner basis of
$\modRel$. By definition, it can be described by the minimal generators of the
leading module of $\modRel$ along with the associated normal forms. We first
show how to use the knowledge of the monomial basis to compute such normal
forms efficiently. In all this section, we let $\mulmats =
(\mulmat{1},\ldots,\mulmat{\nvars})$ be pairwise commuting matrices in
$\matRing[\vsdim]$, $\evMat$ be in $\matRing[\rdim][\vsdim]$, $\ord$ be a
monomial order on $\relSpace$, and $\monbas$ be the \(\ord\)-monomial basis of
\(\relSpace/\modRel\).

\subsubsection{Simultaneous computation of normal forms of monomials}
\label{sec:syzygy:relbas:nf_rref}

First, for some arbitrary monomials $\{\vars^{\expnts_1} \coordVec{i_1},
\ldots, \vars^{\expnts_\ngens} \coordVec{i_\ngens}\}$, we give an algorithm
that computes their $\ord$-normal forms with respect to $\modRel$. Each of
these $\ord$-normal forms is a uniquely defined $\field$-linear combination of
the monomials in the $\ord$-monomial basis $\monbas$; the main task of our
algorithm is to find the coefficients of these \(\ngens\) combinations, which
we gather below in an \(\ngens\times\dimvs\) matrix \(\nfMat\).

In the linearized viewpoint, we associate the monomials $\{\vars^{\expnts_1}
\coordVec{i_1}, \ldots, \vars^{\expnts_\ngens} \coordVec{i_\ngens}\}$ with a
matrix $\termMat \in \matRing[\ngens][\vsdim]$, whose $j$th row is
$\vars^{\expnts_j} \coordVec{i_j} \cdot \evMat = \row{f}_{i_j}
\mulmats^{\expnts_j}$, where $\row{f}_{i_j}$ is the $i_j$th row of $\evMat$ (we
are using the notation of~\cref{dfn:syzygy_module}). Similarly, we associate the
monomial basis $\monbas$ with a matrix $\basMat \in \matRing[\dimvs][\vsdim]$,
where $\dimvs=\dim_\field (\relSpace / \modRel)$ as above; if we write
$\monbas=(\basVec{1} \ord \cdots \ord \basVec{\dimvs})$, then the $j$th row of
$\basMat$ is $\basVec{j} \mul \evMat$. 

Then the rows of $\termMat$ are $\field$-linear combinations of the rows of
$\basMat$: there is a matrix $\nfMat \in \matRing[\ngens][\dimvs]$ such that
$\termMat = \nfMat \basMat$. (The notation $\termMat$ stands for \emph{terms},
while $\basMat$ stands for \emph{basis}, and $\nfMat$ for \emph{normal forms}.)
To compute this matrix $\nfMat$, one can directly use $\nfMat = \termMat
\basMat^{-1}$ if $\basMat$ is square, and more generally one can use a similar
identity $\nfMat = \matt{\termMat} \matt{\basMat}^{-1}$ where $\matt{\basMat}$
is a \(\dimvs\times\dimvs\) invertible submatrix of \(\basMat\), as detailed in
Step~\textbf{1} of \cref{algo:normal_forms}.

\begin{algobox}
  \algoInfo
  {NormalForm}
  {algo:normal_forms}

  \dataInfos{Input}{
    \item matrix $\termMat \in \matRing[\ngens][\vsdim]$,
    \item list of monomials $\monbas = (\basVec{1},\dots,\basVec{\dimvs})$ in $\relSpace$,
    \item matrix $\basMat \in \matRing[\dimvs][\vsdim]$ with full row rank. }

  \dataInfo{Output}{%
    list $(\row{\nu}_1,\ldots,\row{\nu}_\ngens)$ of elements of \(\relSpace\) }

  \dataInfo{Ensures}{%
    assuming the following holds:
    \begin{itemize}
      \item $\mulmats = (\mulmat{1},\ldots,\mulmat{\nvars})$ are pairwise commuting matrices in $\matRing[\vsdim]$,  
      \item $\evMat$ is a matrix in $\evSpace{\vsdim}$ with rows \(\row{f}_1,\ldots,\row{f}_\rdim\),
      \item the $j$th row of $\termMat$ is $\vars^{\expnts_j} \coordVec{i_j} \cdot \evMat = \row{f}_{i_j} \mulmats^{\expnts_j}$
         for some monomial $\vars^{\expnts_j} \coordVec{i_j}$, 
      \item $\monbas$ is the $\ord$-monomial basis of $\relSpace / \modRel$ for some monomial order $\ord$ on $\ring^{\rdim}$,
      \item the $j$th row of $\basMat$ is $\basVec{j} \mul \evMat$,
    \end{itemize}
    then $\row{\nu}_j = \nf{\vars^{\expnts_j} \coordVec{i_j}}$ for \(1 \le j \le \ngens\). }

  \algoSteps{
    \item \comment{Compute the matrix $\nfMat \in \matRing[\ngens][\dimvs]$ such that $\termMat = \nfMat \basMat$} \\
      $(\bar\rkprof_1,\ldots,\bar\rkprof_\dimvs) \assign$ the column rank profile of $\basMat$ \\
      $\matt{\basMat}$ and $\matt{\termMat} \assign$ submatrices of $\basMat$ and $\termMat$ formed by the columns $\{\bar\rkprof_1,\ldots,\bar\rkprof_\dimvs\}$ \\
      $\nfMat = [\nu_{i,j}]_{1 \le i \le s, 1 \le j \le \dimvs} \assign \matt{\termMat} \matt{\basMat}^{-1}$
    \item \comment{Deduce normal forms} \\
      \algoword{For} \(i\) \algoword{from} \(1\) \algoword{to} \(\ngens\):
      \(\row{\nu}_i \assign \nu_{i,1} \basVec{1} + \cdots + \nu_{i,\dimvs} \basVec{\dimvs} \in \relSpace\)
    \item \algoword{Return} $(\row{\nu}_1,\ldots,\row{\nu}_\ngens)$ }
\end{algobox}

\begin{proposition}
  \label{prop:algo:normal_forms}
  \cref{algo:normal_forms} is correct and uses \(\bigO{\dimvs^{\expmatmul-1} (\vsdim
  + \ngens)}\) operations in $\field$.
\end{proposition}
\begin{proof}
  For the correctness, we focus on the case \(\ngens=1\); to prove the general
  case \(\ngens\ge1\) it suffices to apply the following arguments for each \(j
  \in \{1,\ldots,\ngens\}\), to the $j$th row of \(\termMat\) and the
  corresponding output element \(\row{\nu}_j\).  Thus, we consider $\termMat =
  \vars^{\expnts} \coordVec{i} \mul \evMat = \row{f}_i \mulmats^{\expnts}$ for
  some $\expnts \in \NN^\nvars$ and \(1 \le i \le \rdim\), and our goal is to
  prove that \(\nf{\vars^{\expnts} \coordVec{i}} = \nu_1 \basVec{1} + \cdots +
  \nu_\dimvs \basVec{\dimvs}\) where $\nfMat = [\nu_1 \;\; \cdots \;\;
  \nu_\dimvs]$ is the unique vector in $\matRing[1][\dimvs]$ such that
  \(\termMat = \nfMat \basMat\). Choosing large enough exponent bounds
  $\maxDegs \in \NNp^\nvars$, such as \(\maxDegs =
  (\max(\expnts,\vsdim)+1,\ldots,\max(\expnts,\vsdim)+1)\), we recall from
  \cref{dfn:multi_krylov} that \(\termMat\) is a row of the multi-Krylov matrix
  \(\krylov{\ord,\maxDegs}\), and from \cref{thm:rrp_monbas} that \(\basMat\)
  is the submatrix of \(\krylov{\ord,\maxDegs}\) formed by the rows
  corresponding to its row rank profile. This proves the existence and
  uniqueness of \(\nfMat\) such that \(\termMat = \nfMat \basMat\).

  We now explain the computation of \(\nfMat\) in Step~\textbf{1}. We use the
  column rank profile of \(\basMat\) as a specific set of column indices
  $\bar\rkprof_1<\cdots<\bar\rkprof_\dimvs$ such that the corresponding
  \(\dimvs\times\dimvs\) submatrix $\matt{\basMat}$ of $\basMat$ is invertible.
  Then, writing $\matt{\termMat} \in \matRing[1][\dimvs]$ for the subvector of
  \(\termMat\) formed by its entries
  $\{\bar\rkprof_1,\ldots,\bar\rkprof_\dimvs\}$, the identity $\termMat =
  \nfMat \basMat$ yields $\matt{\termMat} = \nfMat \matt{\basMat}$, hence
  $\nfMat = \matt{\termMat} \matt{\basMat}^{-1}$.

  Since the $j$th row of $\basMat$ is $\basVec{j} \mul \evMat$, we have
  \(\mat{0} = \termMat - \nfMat \basMat = (\vars^{\expnts} \coordVec{i} -
  \nu_1 \basVec{1} + \cdots + \nu_\dimvs \basVec{\dimvs}) \mul \evMat\), hence
  $\vars^{\expnts} \coordVec{i} - \nu_1 \basVec{1} + \cdots + \nu_\dimvs
  \basVec{\dimvs}$ is in \(\modRel\). Thus
  \[
    \nf{\vars^{\expnts} \coordVec{i}}
  = \nf{\nu_1 \basVec{1} + \cdots + \nu_\dimvs \basVec{\dimvs}}
  = \nu_1 \basVec{1} + \cdots + \nu_\dimvs \basVec{\dimvs};
  \]
  indeed \(\nu_1 \basVec{1} + \cdots + \nu_\dimvs \basVec{\dimvs}\) is already
  in \(\ord\)-normal form as it is a combination of the \(\ord\)-monomial basis
  \(\monbas\) of $\relSpace / \modRel$. This concludes the proof of correctness.

  Concerning the cost bound, Steps~\textbf{2} and~\textbf{3} do not require
  operations in \(\field\). In Step~\textbf{1}, the column rank profile is
  obtained in $\bigO{\dimvs^{\expmatmul-1} \vsdim}$ operations according to
  \cite[Thm.\,2.10]{Storjohann00}; the inversion of $\matt{\basMat}$ costs
  $\bigO{\dimvs^\expmatmul}$; and the multiplication $\matt{\termMat}
  \matt{\basMat}^{-1}$ uses $\bigO{\ngens \dimvs^{\expmatmul-1}}$ operations if
  $\ngens \ge \dimvs$, and $\bigO{\dimvs^\expmatmul}$ otherwise. Since $\dimvs
  \le \vsdim$ we obtain the announced bound.
\end{proof}

\begin{remark}
  \label{rmk:normal_form:border_basis}
  One may observe that \cref{algo:normal_forms} works in the more general case
  where \(\monbas\) is a basis of the \(\field\)-vector space \(\relSpace /
  \modRel\), and thus in particular when \(\monbas\) is the monomial basis
  associated to a border basis of \(\modRel\) which is not necessarily related
  to a monomial order. In that case, each output element $\row{\nu}_j \in
  \relSpace$ is the unique polynomial which is equal to $\vars^{\expnts_j}
  \coordVec{i_j}$ modulo \(\modRel\) and whose monomials are in \(\monbas\).
  Besides, the argument referring to \cref{thm:rrp_monbas} in the proof above
  can be replaced by the fact that, assuming \(\maxDegs\) large enough,
  \(\basMat\) is a submatrix of \(\krylov{\ord,\maxDegs}\) formed by \(\dimvs =
  \rank{\krylov{\ord,\maxDegs}}\) linearly independent rows. The above cost
  bound thus also holds in this more general case, yet one should note that
  using \cref{algo:normal_forms} requires the knowledge of the input matrix
  \(\basMat\), which might be expensive to compute from \(\monbas\) and
  \(\evMat\) depending on the choice of \(\monbas\). In our specific case,
  \cref{algo:monomial_basis} outputs both \(\monbas\) and \(\basMat\).
\end{remark}


\subsubsection{Computing reduced Gr\"obner bases of syzygies}
\label{sec:syzygy:relbas:groebner}

To compute the reduced $\ord$-Gr\"obner basis $\gb$ of $\modRel$, we start by
using \cref{algo:monomial_basis} to find the $\ord$-monomial basis
$\monbas=(\basVec{1},\ldots,\basVec{\dimvs})$, together with the matrix
$\mat{B}$ giving all $\basVec{k} \mul \evMat$. From $\monbas$, we deduce the set
$\LM = \{ \lt{\rel} \mid \rel \in \gb\}$ formed by the $\ord$-leading terms of
the polynomials in $\gb$, as explained in the next paragraph. Finally, having
$\LM$, we compute $\ord$-normal forms modulo $\modRel$ using
\cref{algo:normal_forms} so as to obtain $\gb = \{f - \nf{f} \mid f \in \LM\}$. We
refer to \cref{sec:preliminaries} for more details concerning the latter
identity and the sets of monomials $\expSet$ and $\border$ used in the next
paragraph.

To find $\LM$, we start from $\monbas = (\basVec{1},\ldots,\basVec{\dimvs})$
and consider the set of multiples
\[
  \expSet = \{ \var_k \basVec{j} \mid 1 \le k \le \nvars, 1 \le j \le \dimvs \}
  \cup \{ \coordVec{i} \mid 1 \le i \le \rdim \text{ such that } \coordVec{i} \not\in \monbas\}.
\]
It gives us the border, as $\border=\expSet-\monbas$. The latter is a set of
monomial generators for the monomial submodule $\genBy{\ltMod{\modRel}}$, while
$\LM$ is the minimal generating set of the same submodule. Thus $\LM$ can be
found from $\border$ by discarding all monomials in $\border$ which are
divisible by another monomial in $\border$.  The number of generators
$\card{\gb}$ is not known in advance; it is at least $\rdim$, and at most
$\nvars\dimvs + \rdim$ as explained in \cref{sec:preliminaries}. In particular,
the output Gr\"obner basis is represented using $\nvars \vsdim^2 + \rdim
\vsdim$ field elements, as claimed in the introduction.

The $\ord$-normal forms of the monomials in $\LM$ will be computed using
\cref{algo:normal_forms}; for this, we need to know $f \mul \evMat$, for $f$ in
$\LM$. The matrix $\mat{B}$ describes all $\basVec{j} \mul \evMat$, for $1 \le
j \le \dimvs$; on the other hand, we know that any $f$ in $\LM$ which is not
among $\{\coordVec{1},\ldots,\coordVec{\rdim}\}$ is a product of the form $f= X_k
\basVec{j}$, for some $k$ in $\{1,\dots,\nvars\}$ and $j$ in
$\{1,\dots,\dimvs\}$. In such a case, $f \mul \evMat$ can be computed as
$(\basVec{j} \mul \evMat) \mulmat{k}$; in the algorithm, we use fast matrix
multiplication to compute several $f \mul \evMat$ at once. Altogether,
\cref{algo:syzygy_module_basis,prop:algo:syzygy_module_basis} prove \cref{thm:grb}.

\begin{algobox}
  \algoInfo
  {SyzygyModuleBasis}
  {algo:syzygy_module_basis}

  \dataInfos{Input}{
    \item monomial order $\ord$ on $\mvPolRing{\nvars}^{1\times\rdim}$,
    \item pairwise commuting matrices $\mulmats = (\mulmat{1},\ldots,\mulmat{\nvars})$ in $\matRing[\vsdim]$, 
    \item matrix $\evMat \in \evSpace{\vsdim}$. }

  \dataInfo{Output}{ the reduced $\ord$-Gr\"obner basis of $\modRel$. }

  \algoSteps{
    \item \comment{Compute monomial basis} \\
      $\monbas=(\basVec{1},\ldots,\basVec{\dimvs}),\basMat \assign \algoname{MonomialBasis}(\ord,\mulmats,\evMat)$ 
    \item \comment{Compute leading monomials and their linearizations} \\
      $\LM \assign$ minimal generating set of $\genBy{\ltMod{\modRel}}$, deduced from $\monbas$ \\[0.1cm] 
      $\LM_0 \assign \LM \cap \{\coordVec{i} \mid 1\le i\le \rdim\}$ \\[0.1cm]
      write $\LM_0$ as  $\{\coordVec{i_{0,j}} \mid 1 \le j \le \ngens_0\}$
      for some indices $i_{0,1},\ldots,i_{0,\ngens_0}$ \\[0.1cm]
      $(\expnts_{0,1}, \ldots, \expnts_{0,\ngens_0}) \assign (0,\ldots,0)$ \\[0.1cm]
      \algoword{For} $k$ \algoword{from} $1$ \algoword{to} $\nvars$ \\[-0.6cm]
      \begin{algosteps}[{\bf a.}]
        \item $\LM_k \assign \{f \in \LM - (\LM_0\cup\cdots\cup\LM_{k-1}) \mid X_k \text{ divides } f
          \text{ and } X_k^{-1} f \in \monbas\}$
        \item write $\LM_k$ as $\{\vars^{\expnts_{k,j}} \coordVec{i_{k,j}} \mid 1 \le j \le \ngens_k\}$
          for some exponents and indices $\expnts_{k,j},i_{k,j}$
        \item \algoword{For} $j$ \algoword{from} $1$ \algoword{to} $\ngens_k$: $\mu_j \assign$ index such that $ X_k^{-1}\vars^{\expnts_{k,j}} \coordVec{i_{k,j}} = \basVec{\mu_j}$
        \item $\termMat_k \assign$ matrix formed by the rows $\mu_1,\ldots,\mu_{\ngens_k}$ of $\basMat$, in this order
        \item $\termMat_k \assign \termMat_k \mulmat{k}$
      \end{algosteps} 
      $\termMat \assign \left[\begin{smallmatrix}\termMat_0 \\ \termMat_1 \\ \svdots \\ \termMat_\nvars\end{smallmatrix}\right]$
    \item \comment{Compute normal forms $\row{\nu}_{k,j} = \nf{\vars^{\expnts_{k,j}} \coordVec{i_{k,j}}}$ and return} \\
      $(\row{\nu}_{0,1},\ldots,\row{\nu}_{0,\ngens_0},\ldots,\row{\nu}_{\nvars,1},\ldots,\row{\nu}_{\nvars,\ngens_\nvars}) \assign \algoname{NormalForm}(\termMat,\monbas,\basMat)$ \\[0.1cm]
      \algoword{Return} $\{\vars^{\expnts_{k,j}} \coordVec{i_{k,j}} - \row{\nu}_{k,j} \mid 1 \le j \le \ngens_k, 0\le k\le \nvars\}$ }
\end{algobox}

\begin{proposition}
  \label{prop:algo:syzygy_module_basis}
  \cref{algo:syzygy_module_basis} is correct and uses
  \[
    \bigO{\rdim \vsdim^{\expmatmul-1} +  \vsdim^\expmatmul (\nvars + \log(d_1 \cdots d_\nvars))} 
    \subset \bigO{\rdim \vsdim^{\expmatmul-1} + \nvars \vsdim^\expmatmul \log(\vsdim)}
  \]
  operations in $\field$, where $d_k \in \{1,\ldots,\vsdim\}$ is the degree
  of the minimal polynomial of $\mulmat{k}$, for $1 \le k \le \nvars$.
\end{proposition}
\begin{proof}
  Concerning correctness, the construction of $\basMat$ ensures that after
  Step~\textbf{2.d}, the rows of $\termMat_k$ are the rows
  $X_k^{-1}\vars^{\expnts_{k,j}} \coordVec{i_{k,j}} \cdot \evMat$. Therefore,
  after Step~\textbf{2.e} the rows of $\termMat_k$ are the rows
  $\vars^{\expnts_{k,j}} \coordVec{i_{k,j}} \cdot \evMat = \row{f}_{i_{k,j}}
  \mulmats^{\expnts_{k,j}}$. Then \cref{prop:algo:normal_forms} implies that
  $\row{\nu}_{k,j}$ computed at Step~\textbf{3} is the normal form
  $\nf{\vars^{\expnts_{k,j}} \coordVec{i_{k,j}}}$, for $1 \le j \le \ngens_k$
  and $0\le k\le \nvars$. This shows the correctness of the algorithm since, as
  explained above, the reduced $\ord$-Gr\"obner basis of syzygies is $\{f -
  \nf{f} \mid f \in \LM\}$.

  Concerning the cost bound, the cost of Step~\textbf{1} is given in
  \cref{prop:algo:monomial_basis} and is precisely the cost bound in the
  present proposition. Then, at the iteration $k$ of the \algoword{For} loop,
  the multiplication at Step~\textbf{2.e} involves the $\ngens_k \times \vsdim$
  matrix $\termMat_k$ and the $\vsdim\times\vsdim$ matrix $\mulmat{k}$.  For
  $k$ in $\{1,\dots,\nvars\}$, since we have by definition $\ngens_k =
  \card{\LM_k} \le \card{\monbas} = \dimvs \le \vsdim$, this multiplication is
  performed in $\bigO{\vsdim^\expmatmul}$ operations; over the $\nvars$
  iterations, this leads to a total of $\bigO{\nvars \vsdim^{\expmatmul}}$
  operations. For $k=0$, we have  $\ngens_0 =\rdim$, so the cost is
  $\bigO{\rdim \vsdim^{\expmatmul-1} + \vsdim^\expmatmul}$.  Finally, we have
  \(\LM_0 \cup \cdots \cup \LM_\nvars = \LM\) and therefore
  \(\ngens_0+\cdots+\ngens_\nvars \le \card{\LM} \le \nvars \dimvs + \rdim \le
  \nvars \vsdim + \rdim\); hence the cost for computing normal forms at
  Step~\textbf{3} is in $\bigO{\dimvs^{\expmatmul-1}
    (\vsdim+\ngens_0+\cdots+\ngens_\nvars)} \subseteq \bigO{\rdim
  \vsdim^{\expmatmul-1} + \nvars \vsdim^\expmatmul}$ according to
  \cref{prop:algo:normal_forms}.
\end{proof}

\begin{remark}
  \label{rmk:syzygy_border_basis}
  By considering \(\border\) instead of \(\LM\) at the second step, one could
  slightly modify \cref{algo:syzygy_module_basis} so that, instead of the reduced
  \(\ord\)-Gr\"obner basis, it returns the border basis with respect to the
  \(\ord\)-monomial basis computed at the first step. One can verify that the
  computation of that monomial basis remains the most expensive step of the
  modified algorithm, and thus that the overall cost bound is the same as the
  one in \cref{prop:algo:syzygy_module_basis}.
\end{remark}

\section{Computing multiplication matrices from the Gr\"obner basis}
\label{sec:mulmat}

In this section, we tackle \cref{pbm:mulmat} and prove \cref{thm:mulmat}. In
what follows, \(\nodule\) is a submodule of \(\ring^\sdim\) such that
$\ring^\sdim/\nodule$ has finite dimension $\vsdim$ as a $\field$-vector space,
$\ord$ is a monomial order on $\ring^\sdim$, and $\gb$ is the reduced
$\ord$-Gr\"obner basis of $\nodule$. Having as input \(\gb\), we give an
algorithm to compute the multiplication matrices for this quotient with respect
to the \(\ord\)-monomial basis, under the assumption on the leading module of
\(\nodule\) described in \cref{dfn:structure_monomial}. For conciseness,
hereafter this assumption is called the \emph{structural assumption}.

\subsection{Overview of the algorithm}
\label{sec:mulmat:overview}

We first discuss the shape of the $\ord$-monomial basis of
$\ring^\sdim/\nodule$ (see \citep{FaGiLaMo93} and \cite[Sec.\,3]{MaMoMo93} for
similar observations in the case of ideals), and then we present an overview of
our approach for computing the multiplication matrices.

Let $\monbas = (\basVec{1},\ldots,\basVec{\vsdim})$ be the $\ord$-monomial
basis of $\ring^\sdim/\nodule$. Below, when discussing multiplication matrices
(with respect to \(\monbas\)) and elements of $\ring^\sdim/\nodule$ represented
on the basis \(\monbas\), we assume that one has fixed an order on the elements
of this basis, for example \(\basVec{1} \ord \cdots \ord \basVec{\vsdim}\).
Then the sought multiplication matrices
$\mulmat{1},\ldots,\mulmat{\nvars}\in\mulmatSpace$ are such that the row $j$ of
$\mulmat{k}$ is the coefficient vector of the normal form $\nf{X_k \basVec{j}}$
represented in the basis $\monbas$, for $1\le k\le \nvars$ and $1 \le j \le
\vsdim$. Thus, we consider the set of these monomials obtained by multiplying
those of $\monbas$ by a variable:
\[
  \expSet = \{ \var_k \basVec{j} \mid 1 \le k \le \nvars, 1 \le j \le \vsdim \}
  \cup
  \{\coordVec{i} \mid 1 \le i \le \sdim \text{ such that } \coordVec{i} \in \ltMod{\nodule}\} .
\]
Note that we have added the coordinate vectors \(\coordVec{i}\) that are in the
leading module \(\ltMod{\nodule}\), or equivalently that are not in \(\monbas\).
This is because the normal forms of these coordinate vectors will also be
computed by our algorithm, for a negligible cost since they will be directly
obtained from the \(\ord\)-Gr\"obner basis.

Regarding the computation of the normal forms of the monomials in $\expSet$, we
can divide them into three disjoint categories:
\[
  \expSet = (\expSet - \border) \disuni \LM \disuni (\border - \LM),
\]
where $\border = \expSet - \monbas$ is the border and $\LM = \lt{\gb} \subseteq
\border$ is the minimal generating set of $\genBy{\ltMod{\nodule}}$ (see
\cref{sec:preliminaries} for more details, and \cref{fig:staircase_2vars} for
an example).

\begin{figure}[t]
  \centering
  \begin{tikzpicture}[
    scale=0.5,
    lm/.style = {rectangle,inner sep=2pt,fill,draw},
    border/.style = {diamond,inner sep=1.4pt,fill,draw},
    expset/.style = {circle,inner sep=1.2pt,draw},
    ]
    \draw[step=1.0,black,very thin] (0.5,0.5) grid (15.5,15.5);
    \draw[thick,->] (0.5,1) -- (16,1) node[below,pos=0.98] {$X$};
    \draw[thick,->] (1,0.5) -- (1,16) node[left,pos=0.98] {$Y$};
    \node[style=lm] at (1,14) {};
    \node[style=lm] at (2,12) {};
    \node[style=lm] at (3,11) {};
    \node[style=lm] at (4,9) {};
    \node[style=lm] at (5,6) {};
    \node[style=lm] at (6,5) {};
    \node[style=lm] at (7,3) {};
    \node[style=lm] at (8,1) {};
    \node[style=border] at (2,13) {};
    \node[style=border] at (4,10) {};
    \node[style=border] at (5,8) {};
    \node[style=border] at (5,7) {};
    \node[style=border] at (7,4) {};
    \node[style=border] at (8,2) {};
    \foreach \j in {2,...,13} { \node[style=expset] at (1,\j) {}; }
    \foreach \j in {1,...,11} { \node[style=expset] at (2,\j) {}; }
    \foreach \j in {1,...,10} { \node[style=expset] at (3,\j) {}; }
    \foreach \j in {1,...,8} { \node[style=expset] at (4,\j) {}; }
    \foreach \j in {1,...,5} { \node[style=expset] at (5,\j) {}; }
    \foreach \j in {1,...,4} { \node[style=expset] at (6,\j) {}; }
    \foreach \j in {1,...,2} { \node[style=expset] at (7,\j) {}; }
      \draw[draw=black,fill=black,opacity=0.1,]
        (7.7,1) -- (7.7,2.7) --
        (6.7,2.7) -- (6.7,4.7) --
        (5.7,4.7) -- (5.7,5.7) --
        (4.7,5.7) -- (4.7,8.7) --
        (3.7,8.7) -- (3.7,10.7) --
        (2.7,10.7) -- (2.7,11.7) --
        (1.7,11.7) -- (1.7,13.7) --
        (1,13.7) -- (1,14.7) --
        (1,14.7) -- (1,15.2) --
        (15.2,15.2) -- (15.2,1) --
        cycle ;
  \end{tikzpicture}
  \caption{Illustration of the sets of exponents in the case of the bivariate
    monomial ideal generated by \(\LM = \{Y^{13}, Y^{11} X, Y^{10} X^2, Y^8
    X^3, Y^5 X^4, Y^4 X^5, Y^2 X^6, X^7\}\). The elements of \(\LM\) are
    represented by squares, those of \(\border - \LM\) by diamonds, and those
  of $\expSet - \border$ by circles. Here, the monomial basis is $\monbas =
\{1\} \cup (\expSet-\border)$. Monomials in the greyed area are those in
\(\genBy{\LM}\), or in other words, those not in $\monbas$.}
\label{fig:staircase_2vars}
\end{figure}

The first set $\expSet - \border$ is contained in $\monbas$; precisely,
\[
  \monbas \,=\,
  (\expSet-\border) \,
  \cup \,
  \{\coordVec{i} \mid 1 \le i \le \sdim \text{ such that } \coordVec{i} \not\in \ltMod{\nodule}\}.
\]
As a result, each monomial in $\expSet-\border$ is its own $\ord$-normal form,
and the corresponding rows of the multiplication matrices are coordinate
vectors of length \(\vsdim\) which are obtained for free.

The monomials in the second set $\LM$ are the \(\ord\)-leading terms of the
elements of \(\gb\), so that $\gb = \{f - \nf{f} \mid f \in \LM\}$. Thus, from
the knowledge of \(\gb\), one can obtain $\nf{\LM}$ using at most $\ngens
\vsdim$ computations of opposites in $\field$, where \(\ngens=\card{\gb}\); by
opposite, we mean having on input $\alpha \in \field$ and computing $-\alpha$.
We recall from \cref{sec:preliminaries} that \(\ngens = \card{\LM} \le
\card{\expSet} < \nvars \vsdim + \sdim\) (the bound is strict here since
\(\vsdim>0\)).

Thus, to obtain the multiplication matrices, the main task is to compute the
normal forms of the third set $\border-\LM$. As discussed above, our algorithm
works under the structural assumption \(\hyp{\genBy{\ltMod{\nodule}}}\) from
\cref{dfn:structure_monomial}. The next lemma summarizes the above discussion
about the computation of \(\nf{\monbas \cup \LM}\), and also highlights one
example of how one can exploit the structural assumption; note that this result
appears in \cite[Sec.\,7]{FaGaHuRe13} in the case of ideals, under slightly
different assumptions.

\begin{lemma}
  \label{lem:bivariate_mulmaty}
  Given the reduced $\ord$-Gr\"obner basis \(\gb\) of \(\nodule\), one can
  compute \(\nf{\monbas \cup \LM}\) using at most $\ngens \vsdim$ operations in
  \(\field\), where \(\ngens\) is the cardinality of \(\gb\). Assuming
  $\hyp{\genBy{\ltMod{\nodule}}}$, we have \(\{X_\nvars\basVec{j} \mid 1 \le j
  \le \vsdim \} \subset \monbas \cup \LM\) and thus $\mulmat{\nvars}$ can be
  read off from \(\nf{\monbas \cup \LM}\).
\end{lemma}
\begin{proof}
  The first claim follows from the above discussion, recalling that \(\ngens\)
  is also the cardinality of \(\LM\). Indeed, \(\nf{\monbas}\) is obtained for
  free, and for each monomial \(f\) in \(\LM\), its normal form \(\nf{f}\) is
  computed using at most \(\vsdim\) computations of opposites of elements of
  \(\field\).

  Now, assume $\hyp{\genBy{\ltMod{\nodule}}}$ and suppose by contradiction that
  $X_\nvars\basVec{j} \not\in \monbas \cup \LM$ for some $j$. Since
  $X_\nvars\basVec{j}$ is not in $\monbas$, it is in $\ltMod{\nodule}$, and thus
  it is a multiple $X_\nvars\basVec{j} = X_1^{\alpha_1} \cdots
  X_\nvars^{\alpha_\nvars} f$ of some $f \in \LM$, for some exponents
  $\alpha_1,\ldots,\alpha_\nvars$. Since $X_\nvars\basVec{j} \not\in \LM$, we
  have $\alpha_k > 0$ for some \(1 \le k\le \nvars\). If $\alpha_\nvars>0$,
  then $\basVec{j} = X_1^{\alpha_1} \cdots X_\nvars^{\alpha_\nvars-1} f \in
  \ltMod{\nodule}$, which is absurd since $\basVec{j} \in \monbas$; hence \(k <
  \nvars\) and $X_\nvars\basVec{j} = X_1^{\alpha_1} \cdots
  X_{\nvars-1}^{\alpha_{\nvars-1}} f$. But then $\frac{1}{X_k}
  X_\nvars\basVec{j} \in \ltMod{\nodule}$, and using
  $\hyp{\genBy{\ltMod{\nodule}}}$ we arrive at the same contradiction:
  $\frac{X_k}{X_\nvars} \frac{1}{X_k} X_\nvars\basVec{j} = \basVec{j} \in
  \ltMod{\nodule}$. Therefore $X_\nvars\basVec{j} \in \monbas \cup \LM$ for all
  $j$, which proves the inclusion in the lemma.

  The last claim follows, since each row of \(\mulmat{\nvars}\) is the
  \(\ord\)-normal form of a monomial $X_\nvars\basVec{j}$, for some \(1 \le j
  \le \vsdim\).
\end{proof}

Computing the remaining multiplication matrices requires to compute normal
forms of monomials in the third set $\border-\LM$, which is more involved. Our
main algorithmic ingredient to compute those efficiently is a procedure which
computes a collection of vector-matrix products of the form \(\row{v}
\mulmats^e\), where \(\mulmats\) is some \(\vsdim\times\vsdim\) matrix; in our
context \(\mulmats\) is one of the multiplication matrices that are already
known at some point of the algorithm. We call this operation \emph{Krylov
evaluation} and we give an algorithm for it in
\cref{sec:mulmat:krylov_evaluation}. The next example gives a simple
illustration of how Krylov evaluation occurs in the computation of
multiplication matrices.

\begin{remark}
  \label{rmk:bivariate_mulmatx}
  Assume \(\nodule\) is an ideal of \(\ring=\field[X_1,X_2]\), that is,
  \(\nvars=2\) and \(\sdim=1\). The above lemma shows how to compute
  $\mulmat{2}$ under the structural assumption. Having \(\mulmat{2}\), we will
  now see how Krylov evaluation allows us to compute \(\mulmat{1}\) using
  $\bigO{\vsdim^\expmatmul \log(\vsdim)}$ operations in $\field$; hence, in
  this context, both multiplication matrices are obtained in this cost bound.

  As explained above, the rows of $\mulmat{1}$ that correspond to normal forms
  in $\nf{\monbas \cup \LM}$ are found using $\bigO{\vsdim^2}$ operations,
  since here \(\ngens\le\vsdim\). Thus, it remains to compute its rows that are
  in $\nf{\border_1}$, where \(\border_1 = \{X_1\basVec{j} \mid 1 \le j \le
  \vsdim\} - (\monbas\cup\LM)\).  Write $\LM = \{ X_1^{\alpha_j} X_2^{\beta_j}
\mid 1\le j \le \ngens\}$ with $(\alpha_{j+1},\beta_{j+1}) \ordLex
(\alpha_j,\beta_j)$ for $1\le j < \ngens$; since \(\LM\) is the minimal
generating set of $\genBy{\ltMod{\nodule}}$, $(\alpha_j)$ is decreasing with
$\alpha_\ngens = 0$ and $(\beta_j)$ is increasing with $\beta_1 = 0$.  Then \(
\border_1 = \{X_1^{\alpha_j} X_2^{\beta_j+k}  \mid 1 \le k < \beta_{j+1} -
\beta_j, 1\le j < \ngens\}.  \)

  Now let $\row{v}_j \in \matRing[1][\vsdim]$ be the vector which represents
  $\nf{X_1^{\alpha_j}X_2^{\beta_j}}$ in the basis $\monbas$; since
  $\{\row{v}_1,\ldots,\row{v}_\ngens\}$ represent $\nf{\LM}$, these vectors are
  among the rows of $\mulmat{1}$ that have already been computed. Then the
  vectors representing $\nf{\border_1}$ are
  \[
    \{\row{v}_j \mulmat{2}^k \mid 1 \le k < \beta_{j+1} - \beta_j, 1\le j < \ngens\}.
  \]
  Performing this Krylov evaluation using the algorithm in
  \cref{sec:mulmat:krylov_evaluation} takes $\bigO{\vsdim^\expmatmul
  \log(\vsdim)}$ operations in $\field$, as stated in
  \cref{lem:krylov_evaluation} which involves parameters that are here $\sigma
  = \beta_{\ngens} \le \vsdim$ and \(\mu \le \vsdim\).
\end{remark}

More generally, for \(\nvars\) variables and \(\sdim\ge1\), one may similarly
obtain $\mulmat{\nvars-1}$ by computing such Krylov evaluations, assuming
$\hyp{\genBy{\ltMod{\nodule}}}$ (this is a consequence of the more general
\cref{lem:read_mulmat,lem:nextvar_iteration}). However, when \(\nvars>2\)
this does not extend into an iterative computation of the multiplication
matrices: $\mulmat{\nvars-2}$ cannot be obtained simply by Krylov evaluation
with the matrix $\mulmat{\nvars-1}$ and the normal forms in
$\nf{\monbas\cup\LM}$ and those given by the rows of $\mulmat{\nvars-1}$. The
reason is that some of the normal forms which constitute the rows of
$\mulmat{\nvars-2}$ are actually obtained by Krylov evaluation with the matrix
$\mulmat{\nvars}$ and the normal forms in $\nf{\LM}$.

Thus we change our focus, from the computation of the multiplication matrices
to that of the normal forms which we can obtained by Krylov evaluation with the
known multiplication matrices and known normal forms. Roughly, our algorithm is
as follows. The first iteration is for \(i=\nvars\) and considers
$\expSet_\nvars = \monbas \cup \LM$, for which we have seen how to efficiently
compute $\nf{\expSet_\nvars}$. Our structural assumption ensures that these
normal forms contain those giving $\mulmat{\nvars}$. Then the iteration
\(i=\nvars-1\) considers the monomials $\expSet_{\nvars-1}$ that can be
obtained from $\expSet_\nvars = \monbas \cup \LM$ by multiplication by
$\var_\nvars$, and their normal forms $\nf{\expSet_{\nvars-1}}$ are computed
using Krylov evaluation with $\mulmat{\nvars}$ and the vectors representing
\(\nf{\expSet_\nvars}\). Again, our assumption ensures that
$\nf{\expSet_{\nvars-1}}$ gives $\mulmat{\nvars-1}$, but it also contains other
normal forms which correspond to rows of multiplication matrices
$\mulmat{1},\ldots,\mulmat{\nvars-2}$, whose computation is not complete yet.
Then we continue this process until $i=1$: at this stage, we have covered the
whole set of monomials $\expSet$ and we thus have all the normal forms in
$\nf{\expSet}$, from which we read the rows of the multiplication matrices. The
algorithm is described in detail in \cref{sec:mulmat:computing_mulmat}.

\subsection{Algorithm for Krylov evaluation}
\label{sec:mulmat:krylov_evaluation}

Now we give a simple method for the computation of a collection of
vector-matrix products of the form $\row{v} \mulmats^\expnt$, obtaining
efficiency via repeated squaring of the matrix $\mulmats$. This is detailed in
\cref{algo:krylov_evaluation}, in which we use the following conventions. When
specified, instead of indexing the rows of a matrix $\mat{K} \in
\matRing[\sigma][\vsdim]$ using the integers \((1,\ldots,\sigma)\), we index
them by a given totally ordered set $(\mathcal{P},\le)$ of cardinality
$\sigma$.  Explicitly, if \(\mathcal{P}=\{\expnts_1,\ldots,\expnts_\sigma\}\)
with \(\expnts_1 \le \cdots \le \expnts_\sigma\), then the \(i\)th row of
\(\mat{K}\) has index \(\expnts_i\). Then, for any subset $\mathcal{P}'
\subseteq \mathcal{P}$, we write $\matrows{\mat{K}}{\mathcal{P}'}$ for the
submatrix of $\mat{K}$ formed by its rows with indices in $\mathcal{P}'$. An
assignment operation such as $\matrows{\mat{K}}{\mathcal{P}'} \assign \mat{A}$
for some $\mat{A} \in \matRing[\card{\mathcal{P'}}][\vsdim]$ does modify the
corresponding entries of $\mat{K}$.

\begin{algobox}
  \algoInfo
  {KrylovEval}
  {algo:krylov_evaluation}

  \dataInfos{Input}{%
    \item a matrix $\mulmats \in \matRing[\vsdim]$ for some $\vsdim\in\NNp$,
    \item row vectors $\row{v}_1,\ldots,\row{v}_\nbvec \in \matRing[1][\vsdim]$ for some
      $\nbvec \in \NNp$,
    \item bounds $\expnt_1,\ldots,\expnt_\nbvec \in \NNp$.  }

  \dataInfo{Output}{%
    the matrix $\mat{K} \in \matRing[(\expnt_1+\cdots+\expnt_\nbvec)][\vsdim]$
    whose row $\expnt_1+\cdots+\expnt_{j-1} + \expnt$ is equal to $\row{v}_j
    \mulmats^\expnt$, for $1 \le \expnt \le \expnt_j$ and $1\le j\le \nbvec$.  }

  \algoSteps{
    \item \label{step:krylov_eval:init} \comment{Initialize set of indices and output matrix} \\
      $\mathcal{P} \assign \{(\expnt,j) \mid 1 \le \expnt \le \expnt_j, 1\le j\le \nbvec\}$ \\
      $\mat{K} \assign \mat{0} \in \matRing[(\expnt_1 + \cdots + \expnt_\nbvec)][\vsdim]$ with its rows indexed by $(\mathcal{P},\ordLex)$
    \item \label{step:krylov_eval:e_one} \comment{Case $\expnt=1$:~compute $\row{v}_j \mulmats$ for $1\le j\le \nbvec$} \\
      $\mathcal{P}' \assign \{(1,j) \mid 1 \le j \le \nbvec\}$ \eolcomment{$\mathcal{P}' \subseteq \mathcal{P}$} \\
      $\matrows{\mat{K}}{\mathcal{P}'} \assign \trsp{\begin{bmatrix} \trsp{\row{v}_1} & \cdots & \trsp{\row{v}_\nbvec} \end{bmatrix}} \mulmats$
    \item \label{step:krylov_eval:loop} \comment{Repeated squaring:~handle \(\expnt \in \{2^{i-1}+1,\ldots,2^i\}\) for \(i>1\)} \\
      $\mat{N} \assign \mulmats$ \\
     \algoword{For} $i$ \algoword{from} $1$ \algoword{to} $\lceil
      \log_2(\max_i \expnt_i) \rceil$:
      \begin{algosteps}[ ]
        \item \algoword{If} $i>1$ \algoword{then} $\mat{N} \assign \mat{N}^2$
          \eolcomment{\(\mat{N} = \mulmats^{2^{i-1}}\)}
        \item $\mathcal{P}' \assign \{(\expnt,j) \mid 2^{i-1} < \expnt \le \min(\expnt_j,2^i), 1 \le j \le \nbvec\}$
              \eolcomment{$\mathcal{P}' \subseteq \mathcal{P}$}
        \item $\mathcal{P}'' = \{(\expnt-2^{i-1},j) \mid (\expnt,j) \in \mathcal{P}'\}$
              \eolcomment{$\mathcal{P}'' \subseteq \mathcal{P}$}
        \item $\matrows{\mat{K}}{\mathcal{P}'} \assign \matrows{\mat{K}}{\mathcal{P}''} \cdot \mat{N}$
      \end{algosteps}
    \item \algoword{Return} $\mat{K}$ }
\end{algobox}

\begin{lemma}
  \label{lem:krylov_evaluation}
  Given $\mulmats \in \matRing[\vsdim]$, let $\row{v}_1,\ldots,\row{v}_\nbvec
  \in \matRing[1][\vsdim]$, and let $\expnt_1,\ldots,\expnt_\nbvec \in \NNp$,
  \cref{algo:krylov_evaluation} computes the row vectors $\{\row{v}_j
  \mulmats^\expnt \mid 1 \le \expnt \le \expnt_j, 1\le j\le \nbvec\}$ using 
  \[
    \bigO{\vsdim^\expmatmul (1+\log(\mu)) + \vsdim^{\expmatmul-1} \sigma (1 + \log(\mu\nbvec/\sigma))}
    \;\subseteq\;
    \bigO{\vsdim^{\expmatmul-1} (\vsdim+\sigma) (1+\log(\mu))}
  \]
  operations in $\field$, where $\sigma = \expnt_1 + \cdots + \expnt_\nbvec$
  and $\mu = \max(\expnt_1,\ldots,\expnt_\nbvec)$.
\end{lemma}
\begin{proof}
  We want to prove that after \cref{step:krylov_eval:loop}, the row
  $\expnt_1+\cdots+\expnt_{j-1} + \expnt$ of \(\mat{K}\) is $\row{v}_j
  \mulmats^\expnt$, for $1 \le \expnt \le \expnt_j$ and $1\le j\le \nbvec$.
  Indexing the rows of \(\mat{K}\) by $(\mathcal{P},\ordLex)$ as in the
  algorithm, this means that the row of \(\mat{K}\) at index \((\expnt,j)\) is
  $\row{v}_j \mulmats^\expnt$, for all \((\expnt,j) \in \mathcal{P}\). To prove
  this, we show that at the end of the \(i\)th iteration of the loop at
  \cref{step:krylov_eval:loop}, the following assertion holds:
  \begin{align*}
    \mathcal{A}_i: & \text{ the row of } \mat{K} \text{ at index } (\expnt,j) \text{ is equal to } \row{v}_j \mulmats^\expnt, \\
                   & \text{ for all } (\expnt,j) \in \mathcal{P} \text{ such that } \expnt \le 2^i.
  \end{align*}
  This gives the conclusion, since when the algorithm completes the loop, we
  have $i= \lceil \log_2(\mu) \rceil$ and \(2^i \ge \mu\), and all \((\expnt,j)
  \in \mathcal{P}\) are such that \(\expnt \le \mu\) by definition of \(\mu\).

  First, \((1,j) \in \mathcal{P}\) for $1\le j\le \nbvec$, and after
  \cref{step:krylov_eval:e_one}, the row of \(\mat{K}\) at index \((1,j)\) is
  equal to $\row{v}_j \mulmats$. So \(\mathcal{A}_0\) holds before the first
  iteration \(i=1\). Now, assume that \(\mathcal{A}_{i-1}\) holds before the
  \(i\)th iteration: we want to prove that \(\mathcal{A}_i\) holds at the end
  of this iteration. After the squaring at the beginning of the iteration, we
  have \(\mat{N} = \mulmats^{2^{i-1}}\). By construction, \(\mathcal{P}'\) is the set
  of indices such that we have the equality
  \[
    \mathcal{P}' \disuni
    (\mathcal{P} \cap \{(\expnt,j) \mid 1 \le \expnt \le 2^{i-1}, 1 \le j \le \nbvec\})
    =
    \mathcal{P} \cap \{(\expnt,j) \mid 1 \le \expnt \le 2^i, 1 \le j \le \nbvec\}.
  \]
  Thus our goal is to show that for each \((\expnt,j) \in \mathcal{P}'\), at
  the end of the iteration the row of \(\mat{K}\) at index \((\expnt,j)\) is
  \(\row{v}_j \mulmats^{\expnt}\). By assumption, the row of \(\mat{K}\) at
  index \((\expnt-2^{i-1},j)\) is \(\row{v}_j \mulmats^{\expnt-2^{i-1}}\).
  Then the last step of the iteration ensures that the row of \(\mat{K}\) at
  index \((\expnt,j)\) is \(\row{v}_j \mulmats^{\expnt-2^{i-1}} \mat{N} =
  \row{v}_j \mulmats^{\expnt-2^{i-1}} \mulmats^{2^{i-1}} = \row{v}_j
  \mulmats^{\expnt}\). Thus, \(\mathcal{A}_i\) holds, and this concludes the
  proof of correctness.

  At \cref{step:krylov_eval:e_one}, we multiply an \(\nbvec\times\vsdim\)
  matrix and a \(\vsdim\times\vsdim\) matrix, using $\bigO{\vsdim^{\expmatmul}
  + \vsdim^{\expmatmul-1} \nbvec}$ operations in \(\field\); this is within the
  bound in the lemma since \(\nbvec \le \sigma\). Over all iterations of the loop at
  \cref{step:krylov_eval:loop}, the squarings of \(\mat{N}\) use a total of
  $\bigO{\vsdim^\expmatmul (\lceil\log_2(\mu)\rceil-1)} \subseteq
  \bigO{\vsdim^\expmatmul \log_2(\mu)}$ operations. The product
  \(\matrows{\mat{K}}{\mathcal{P}''} \cdot \mat{N}\) is computed using
  \(\bigO{\vsdim^\expmatmul + \vsdim^{\expmatmul-1} \card{\mathcal{P}''}}\)
  operations in \(\field\) since the matrices have size
  \(\card{\mathcal{P}''}\times\vsdim\) and \(\vsdim\times\vsdim\),
  respectively. By definition, \(\card{\mathcal{P}''} = \card{\mathcal{P}'}\),
  and since \(\mathcal{P}' = \mathcal{P} \cap \{(\expnt,j) \mid 2^{i-1} <
  \expnt \le 2^{i}, 1 \le j \le \nbvec\}\), we obtain \(\card{\mathcal{P}''}
  \le \card{\mathcal{P}} = \sigma\) and \(\card{\mathcal{P}''} \le 2^{i-1} \nbvec\).
  Then the total number of operations in \(\field\) used for the computation
  of \(\matrows{\mat{K}}{\mathcal{P}''} \cdot \mat{N}\) over all iterations of
  the loop at \cref{step:krylov_eval:loop} is
  \[
    \bigO{\sum_{i=1}^{\lceil \log_2(\mu) \rceil} (\vsdim^\expmatmul + \vsdim^{\expmatmul-1} \min(2^{i-1}\nbvec, \sigma))} 
    \subseteq \bigO{\vsdim^\expmatmul \lceil \log_2(\mu) \rceil + \vsdim^{\expmatmul-1} \sum_{i=1}^{\lceil \log_2(\mu) \rceil} \min(2^{i-1}\nbvec, \sigma)},
  \]
  which is within the cost bound in the lemma. Indeed, \(\lceil \log_2(\mu) \rceil \le 1
  + \log_2(\mu)\) and
  \begin{align*}
    \sum_{i=1}^{\lceil \log_2(\mu) \rceil} \min(2^{i-1}\nbvec, \sigma)
    \le \sum_{i=1}^{\lfloor\log_2(\sigma/\nbvec)\rfloor} 2^{i-1} \nbvec + \sum_{i=\lfloor\log_2(\sigma/\nbvec)\rfloor+1}^{\lceil \log_2(\mu) \rceil} \sigma 
    & \le \sigma (1 + \lceil \log_2(\mu) \rceil - \lfloor\log_2(\sigma/\nbvec)\rfloor) \\
    & \le \sigma (3 + \log_2(\mu\nbvec/\sigma)). \qedhere
  \end{align*}
\end{proof}

\subsection{Computing the multiplication matrices}
\label{sec:mulmat:computing_mulmat}

Now, we describe our algorithm for computing the multiplication matrices and
give a complexity analysis. We follow on from notation in
\cref{sec:mulmat:overview}.

We are going to show how to compute the normal forms of all monomials in
\[
  \monbas \cup \border =
  \expSet \cup \{\coordVec{i} \mid 1 \le i \le \sdim \text{ such that } \coordVec{i} \not\in \ltMod{\nodule}\};
\]
since this set contains \(\expSet\), this directly yields the multiplication
matrices. We design an iteration on the \(\nvars\) variables which computes the
normal forms of \(\nvars\) sets \(\monbas \cup \LM = \nextExpSets{\nvars}
\subseteq \nextExpSets{\nvars-1} \subseteq \cdots \subseteq \nextExpSets{1} =
\monbas \cup \border\); at the end, \(\nf{\nextExpSets{1}}\) gives the sought
normal forms.

Thus, we start with the normal forms of the monomials in $\nextExpSets{\nvars}
= \monbas \cup \LM$, which are easily found from $\gb$ (see
\cref{lem:bivariate_mulmaty}). Then, for $1 \le i < \nvars$, we consider the
monomials in $\monbas \cup \border$ which are obtained from $\monbas \cup \LM$
through multiplication by $\var_{i+1},\ldots,\var_\nvars$:
\[
  \nextExpSets{i} = \{\var_{i+1}^{\expnt_{i+1}} \cdots \var_\nvars^{\expnt_\nvars} f \mid
            \expnt_{i+1},\ldots,\expnt_{\nvars} \ge 0, f \in \monbas
          \cup \LM \} \cap (\monbas \cup \border).
\]
The normal forms \(\nf{\nextExpSets{i}}\) can be obtained from those in
$\nf{\monbas \cup \LM}$ through multiplication by
$\mulmat{i+1},\ldots,\mulmat{\nvars}$, if these matrices are known. 

From these sets, we define the sets mentioned at the end of
\cref{sec:mulmat:overview}: $\expSet_{\nvars} = \nextExpSets{\nvars} = \monbas
\cup \LM$ for $i=\nvars$, and $\expSet_i = \nextExpSets{i} - \nextExpSets{i+1}$
for $1 \le i <\nvars$.  Therefore $\nextExpSets{i}$ is the disjoint union
$\expSet_i \disuni \cdots \disuni \expSet_\nvars$, and $\expSet_i$ is the set
of monomials in $\border-\nextExpSets{i+1}$ which can be obtained from $\monbas
\cup \LM$ through multiplication by a monomial in
$\var_{i+1},\ldots,\var_\nvars$ which does involve the variable $\var_{i+1}$. That is,
\begin{align*}
  \expSet_{i} & = \{ \var_{i+1}^{\expnt_{i+1}} \cdots \var_\nvars^{\expnt_\nvars} f \mid
    \expnt_{i+1} > 0, \expnt_{i+2},\ldots,\expnt_{\nvars} \ge 0, f \in \monbas
  \cup \LM \} \cap (\border - \nextExpSets{i+1}) \\
            & = \{ \var_{i+1}^\expnt f \mid \expnt>0, f \in \nextExpSets{i+1} \}
            \cap (\border - \nextExpSets{i+1}).
\end{align*}
In particular, if $\mulmat{i+1}$ and $\nf{\nextExpSets{i+1}}$ are known, then
$\nf{\expSet_{i}}$ can be computed via Krylov evaluation with the matrix
$\mulmat{i+1}$ and the vectors representing $\nf{\nextExpSets{i+1}}$.  Having
$\nf{\expSet_{i}}$ gives us \(\nf{\nextExpSets{i}} = \nf{\expSet_{i}} \cup
\nf{\nextExpSets{i+1}}\), and we will prove in \cref{lem:read_mulmat} that
$\mulmat{i}$ can be read off from $\nf{\nextExpSets{i}}$, under the structural
assumption. Thus we can proceed iteratively, since then, from $\mulmat{i}$ and
$\nf{\nextExpSets{i}}$ we can use Krylov evaluation to find
\(\nf{\expSet_{i-1}}\), from which we deduce $\mulmat{i-1}$, etc. At the end of
this process we have computed \(\nf{\nextExpSets{1}} \supseteq \nf{\expSet}\)
and deduced all the multiplication matrices.

\begin{lemma}
  \label{lem:read_mulmat}
  Assuming \(\hyp{\genBy{\ltMod{\nodule}}}\), we have
  \[
    \{\var_i \basVec{j} \mid 1 \le j \le \vsdim \} \;\;\subseteq\;\; \nextExpSets{i}
    \quad\text{for all}\;\; 1 \le i \le \nvars ;
  \]
  in particular, the multiplication matrices $\mulmat{i}, \ldots,
  \mulmat{\nvars}$ can be read off from $\nf{\nextExpSets{i}}$.
\end{lemma}
\begin{proof}
  Note that for $i=\nvars$, this result was already proved in
  \cref{lem:bivariate_mulmaty} (and we will use similar arguments in the proof
  below); besides, for $i=1$ it is straightforward since $\nextExpSets{1} =
  \monbas \cup \border$.
  
  Let $i \in \{1,\ldots,\nvars\}$. First, $\mulmat{i+1}, \ldots,
  \mulmat{\nvars}$ can be read off from $\nf{\nextExpSets{i}}$ since for $k \in
  \{i+1,\ldots,\nvars\}$ we have $\{\var_{k} \basVec{j}, 1\le j \le\vsdim\}
  \subseteq \nextExpSets{k-1} \subseteq \nextExpSets{i}$. The fact that
  \(\mulmat{i}\) can be read off from $\nf{\nextExpSets{i}}$ follows from the
  inclusion $\{\var_i \basVec{j} \mid 1 \le j \le \vsdim \} \subseteq
  \nextExpSets{i}$ in the lemma; it remains to prove that inclusion.

  Thus, we want to prove $\var_i \basVec{j} \in \nextExpSets{i}$ for any $j \in
  \{1, \ldots, \vsdim\}$; for this we will use the structural assumption.  The
  particular case $\var_i \basVec{j} \in \monbas\cup\LM$ is obvious since
  \(\monbas\cup\LM = \expSet_\nvars \subseteq \nextExpSets{i}\). Now we
  consider $\var_i \basVec{j} \not\in \monbas\cup\LM$. Then $\var_i \basVec{j}
  \in \ltMod{\nodule}$ and there exist exponents $\alpha_1,\ldots,\alpha_r$ not
  all zero and a monomial $f \in \LM$ such that $\var_i \basVec{j} =
  \var_1^{\alpha_1} \cdots \var_\nvars^{\alpha_\nvars} f$.

  Suppose by contradiction that there exists \(k\in\{1,\ldots,i\}\) such that
  $\alpha_k > 0$. If $k=i$, then $\alpha_i > 0$ implies that $\basVec{j}$ is a
  multiple of $f$, hence $\basVec{j} \in \ltMod{\nodule}$, which is not the
  case. If $1 \le k<i$, $\alpha_k > 0$ implies $\frac{1}{\var_k}\var_i
  \basVec{j} \in \ltMod{\nodule}$, and using the structural assumption we
  obtain the same contradiction: $\frac{\var_k}{\var_i} \frac{1}{\var_k}\var_i
  \basVec{j} = \basVec{j} \in \ltMod{\nodule}$. Thus there is no such $k$, and
  $\alpha_1 = \cdots = \alpha_i = 0$. As a result, $\var_i \basVec{j} =
  \var_{i+1}^{\alpha_{i+1}} \cdots \var_\nvars^{\alpha_\nvars} f$, which is in
  \(\nextExpSets{i}\).
\end{proof}

For completeness, in  \cref{algo:next_monomials} we describe a straightforward
subroutine for determining the sets of monomials $(\expSet_i)_{1\le i\le
\nvars}$, which is directly based on the description of these sets in
\cref{lem:expSeti}. Note that this computation does not involve field
operations but only comparisons of exponents of monomials, so that here the
time for finding these sets is not taken into account in our cost bounds. In an
efficient implementation of our algorithm for finding the multiplication
matrices, one would rather compute these sets while building $\border$ from
$\gb$, and we believe that finding these sets should indeed be a negligible
part of the running time of such an implementation.

\begin{lemma}
  \label{lem:expSeti}
  Assume \(\hyp{\genBy{\ltMod{\nodule}}}\) and let \(1 \le i < \nvars\). Let
  \[
    \{f_1,\ldots,f_t\} = \{f \in \nextExpSets{i+1} \cap \border \mid X_{i+1} f
    \in \border - \nextExpSets{i+1}\},
  \]
  and for \(1 \le j \le t\), let \(\expnt_j \in \NNp\) be the largest integer such
  that $\var_{i+1}^{\expnt_j} f_j \in \border - \nextExpSets{i+1}$. Then
  \[
    \expSet_{i} = \{\var_{i+1}^{\expnt} f_j \mid 1\le \expnt \le \expnt_j, 1 \le j \le t\}.
  \]
\end{lemma}
\begin{proof}
  First, we prove that \(\expSet_i\) contains the latter set. Let $1 \le j\le
  t$ and \(1 \le \expnt \le \expnt_j\). Then $f_j$ is in \(\nextExpSets{i+1}\)
  and $\var_{i+1}^{\expnt} f_j$ is in \(\border\). Since $\expnt > 0$,
  $\var_{i+1}^{\expnt} f_j$ is in $\nextExpSets{i}$. Furthermore,
  $\var_{i+1}^{\expnt} f_j$ is not in $\nextExpSets{i+1}$, so that it is in
  \(\nextExpSets{i} - \nextExpSets{i+1} = \expSet_i\).

  Now, let \(g \in \expSet_i\). By definition, this means that $g$ is in
  $\nextExpSets{i} \subseteq \monbas \cup \border$, and $g$ is not in
  \(\nextExpSets{i+1} \supseteq \monbas\cup\LM\). In particular, $g$ is in
  $\border - \nextExpSets{i+1}$. Furthermore, we can write \(g = X_{i+1}^\expnt
  f\) for some \(\expnt > 0\) and \(f \in \nextExpSets{i+1}\).
  Then let \(\expnt'\) be the smallest exponent such that \(X_{i+1}^{\expnt'}f
  \not\in \nextExpSets{i+1}\); we have \(1 \le \expnt' \le \expnt\) since $f
  \in \nextExpSets{i+1}$ and $g \not\in \nextExpSets{i+1}$. Let \(f' =
  X_{i+1}^{\expnt'-1}f \in \nextExpSets{i+1}\); thus \(f'\) is in \(\border\):
  if this was not the case, then \(f'\) would be in \(\monbas\) and \(X_{i+1}
  f'\) would be in \(\nextExpSets{i+1}\) according to \cref{lem:read_mulmat}.
  Furthermore, it is a property of the border that, since the multiple $g =
  X_{i+1}^{\expnt-\expnt'+1} f'$ is in $\border$, then $X_{i+1} f'$ is in
  $\monbas \cup \border$; yet $X_{i+1} f'$ is not in $\nextExpSets{i+1}$ which
  contains $\monbas$, hence \(X_{i+1} f' \in \border-\nextExpSets{i+1}\). It
  follows that \(f' = f_j\) for some \(1 \le j \le t\). Thus we have \(g =
  X_{i+1}^{\expnt-\expnt'+1} f_j\), with \(\expnt-\expnt'+1 \le \expnt_j\) by
  definition of \(\expnt_j\), which concludes the proof.
\end{proof}

\begin{algobox}
  \algoInfo
  {NextMonomials}
  {algo:next_monomials}

  \dataInfos{Input}{
    \item the border \(\border\),
    \item the set of monomials \(\nextExpSets{i+1} = \expSet_{i+1} \cup \cdots \cup \expSet_{\nvars}\) for some \(1 \le i <\nvars\).  }

  \dataInfo{Output}{ the set of monomials $\expSet_i$, in the form
    $\expSet_i = \{\var_{i+1}^{\expnt} f_j \mid 1\le \expnt \le \expnt_j, 1 \le j
    \le t\}$ for some $f_1,\ldots,f_t \in \nextExpSets{i+1} \cap \border$ and
$\expnt_1,\ldots,\expnt_t \in \NNp$. }

  \algoSteps{
    \item \(\expSet_i \assign \emptyset; n \assign 0\)
    \item \algoword{For each} $f \in \nextExpSets{i+1} \cap \border$ such that $X_{i+1} f \in \border - \nextExpSets{i+1}$:
          \begin{algosteps}[ ]
            \item $\expnt \assign 1$;~\algoword{While} $\var_{i+1}^{\expnt+1} f \in \border - \nextExpSets{i+1}$:
              $\expnt \assign \expnt+1$
            \item $\expSet_i \assign \expSet_i \cup \{\var_{i+1} f, \ldots, \var_{i+1}^{\expnt} f\}$
            \item $t \assign t+1; \expnt_t \assign \expnt; f_t \assign f$
          \end{algosteps}
    \item \algoword{Return} $\expSet_i = \{\var_{i+1}^{\expnt} f_j \mid 1\le \expnt \le \expnt_j, 1 \le j \le t\}$ }
\end{algobox}

Next, we show how to compute $\nf{\expSet_i}$ from $\nf{\nextExpSets{i+1}}$ and
$\mulmat{i+1}$ using Krylov evaluation.

\begin{lemma}
  \label{lem:nextvar_iteration}
  Let $i \in \{1,\ldots,\nvars-1\}$. Given
  $(\border,\nextExpSets{i+1},\nf{\nextExpSets{i+1}},\mulmat{i+1})$, one can
  compute $\expSet_i$ and $\nf{\expSet_i}$ as follows:
  \begin{itemize}
    \item $\expSet_i = \{\var_{i+1}^{\expnt} f_j \mid 1\le \expnt \le \expnt_j, 1 \le j \le \nbvec_i\} \assign \algoname{NextMonomials}(\border,\nextExpSets{i+1})$,
      for some $f_1,\ldots,f_{\nbvec_i} \in \nextExpSets{i+1} \cap \border$ and $e_1,\ldots,e_{\nbvec_i} \in \NNp$
    \item \(\{\row{v}_1,\ldots,\row{v}_{\nbvec_i}\} \subseteq \matRing[1][\vsdim] \assign \nf{\{f_1,\ldots,f_{\nbvec_i}\}}\), retrieved from \(\nf{\nextExpSets{i+1}}\) \\
      $\nf{\expSet_i} \assign$ rows of $\algoname{KrylovEval}(\mulmat{i+1},\row{v}_1,\ldots,\row{v}_{\nbvec_i},\expnt_1,\ldots,\expnt_{\nbvec_i})$
  \end{itemize}
  This uses
  \[
    \bigO{\vsdim^\expmatmul (1+\log(\mu_i)) + \vsdim^{\expmatmul-1} \sigma_i (1 + \log(\mu_i\nbvec_i/\sigma_i))}
    \;\subseteq\;
    \bigO{\vsdim^{\expmatmul-1} (\vsdim+\sigma_i) (1+\log(\mu_i))}
  \]
  operations in $\field$, where $\sigma_i = e_1 + \cdots + e_{\nbvec_i}$ is the
  cardinality of $\expSet_i$ and \(\mu_i = \max(e_1,\ldots,e_{\nbvec_i})\). We
  have $\mu_i \le \max \{\expnt \in \NN \mid \var_{i+1}^\expnt f \in \border
    \text{ for some } f \in \border\}$.
\end{lemma}
\begin{proof}
  In the first step, $\expSet_i$ is determined from $\nextExpSets{i+1}$ without
  field operation as shown in \cref{algo:next_monomials}, and it is obtained in the
  form $\expSet_i = \{\var_{i+1}^{\expnt} f_j \mid 1\le \expnt \le \expnt_j, 1
  \le j \le \nbvec_i\}$, where $f_1,\ldots,f_{\nbvec_i}$ are elements of
  $\nextExpSets{i+1}$ and $\expnt_1,\ldots,\expnt_{\nbvec_i}$ are positive
  integers; in particular, $\expnt_1 + \cdots + \expnt_{\nbvec_i} =
  \card{\expSet_i} = \sigma_i$. The upper bound on \(\mu_i\) holds since, by
  construction of \(\expSet_i\) as in \cref{algo:next_monomials} (see also
  \cref{lem:expSeti}), \(f_j \in \border\) and \(X_{i+1}^{e_j} f_j \in
  \border\) for \(1 \le j \le \nbvec_i\).

  Going to the normal forms, we get
  \begin{equation}
    \label{eqn:nf_exptseti}
    \nf{\expSet_i} = \{\row{v}_j \mulmat{i+1}^{\expnt}
   \mid 1\le \expnt \le \expnt_j, 1 \le j \le \nbvec_i\}
  \end{equation}
  where $\row{v}_j = \nf{f_j} \in \matRing[1][\vsdim]$ for \(1\le j\le
  \nbvec_i\); these normal forms are already known since they are in
  \(\nf{\nextExpSets{i+1}}\). This shows that the second item correctly
  computes \(\nf{\expSet_i}\). The cost bound is a consequence of
  \cref{lem:krylov_evaluation}.
\end{proof}

\begin{algobox}
  \algoInfo
  {MultiplicationMatrices}
  {algo:multiplication_matrices}

  \dataInfos{Input}{%
    \item a monomial order $\ord$ on $\ring^\sdim$,
    \item a reduced $\ord$-Gr\"obner basis $\gb \subseteq \ring^\sdim$. }

  \dataInfos{Requirements}{%
    \item \(\ring^\sdim/\nodule\) has finite dimension \(\vsdim\) as a
      $\field$-vector space, where $\nodule = \genBy{\gb}$,
    \item \(\hyp{\genBy{\ltMod{\nodule}}}\) holds. }

  \dataInfos{Output}{%
    \item the \(\ord\)-monomial basis $\monbas$ of \(\ring^\sdim/\nodule\),
    \item the multiplication matrices of $X_1,\ldots,X_\nvars$ in
      $\ring^\sdim/\nodule$ with respect to \(\monbas\). }

  \algoSteps{
    \item \comment{Build main sets of exponents} \\
      Read $\LM$ and $\monbas = (\basVec{1},\ldots,\basVec{\vsdim})$ from $\gb$, with \(\basVec{1} \ord \cdots \ord \basVec{\vsdim}\) \\
      $\expSet \assign \{\var_k \basVec{j} \mid 1 \le k \le \nvars, 1 \le j
              \le \vsdim\}
             \cup
            \{\coordVec{i} \mid 1 \le i \le \sdim \text{ such that } \coordVec{i} \in \LM\}$ \\
      $\border \assign \expSet - \monbas$ \\
      \comment{Below, normal forms are represented in \(\monbas\), as vectors in $\matRing[1][\vsdim]$}
    \item \comment{Initialize the iteration:~\(\nextExpSets{}=\nextExpSets{\nvars}=\monbas\cup\LM\), $\mathcal{Q} = \nf{\nextExpSets{i+1}}$, find \(\mulmat{\nvars}\)} \\
      $\nextExpSets{} \assign \monbas \cup \LM$;~
      $\mathcal{Q} \assign \nf{\nextExpSets{}}$;~
      read $\mulmat{\nvars}$ from $\mathcal{Q}$
    \item \algoword{For} $i$ \algoword{from} $\nvars-1$ \algoword{to} $1$: \\
      \hphantom{sp}\comment{Before iteration $i$:~$\nextExpSets{} = \nextExpSets{i+1}$, $\mathcal{Q} = \nf{\nextExpSets{i+1}}$, \(\mulmat{i+1},\ldots,\mulmat{\nvars}\) known} \\
      \hphantom{sp}\comment{After iteration $i$:~$\nextExpSets{} = \nextExpSets{i}$, $\mathcal{Q} = \nf{\nextExpSets{i}}$, \(\mulmat{i},\ldots,\mulmat{\nvars}\) known}
      \begin{algosteps}[{\bf a.}]
        \item $\widetilde{\expSet} = \{\var_{i+1}^{\expnt} f_j \mid 1\le \expnt \le \expnt_j, 1 \le j \le t\} \assign \algoname{NextMonomials}(\border,\nextExpSets{})$,
          for some $f_1,\ldots,f_t \in \nextExpSets{}$ and $e_1,\ldots,e_t \in \NNp$
          \eolcomment{$\widetilde{\expSet} = \expSet_i$}
        \item \(\{\row{v}_1,\ldots,\row{v}_t\} \subseteq \matRing[1][\vsdim] \assign \nf{\{f_1,\ldots,f_t\}}\), retrieved from \(\mathcal{Q} = \nf{\nextExpSets{}}\) \\
          $\widetilde{\mathcal{Q}} \assign$ rows of $\algoname{KrylovEval}(\mulmat{i+1},\row{v}_1,\ldots,\row{v}_t,\expnt_1,\ldots,\expnt_t)$
          \eolcomment{$\widetilde{\mathcal{Q}}=\nf{\expSet_i}$}
        \item $\nextExpSets{} \assign \widetilde{\expSet} \cup \nextExpSets{}$;~
          $\mathcal{Q} \assign \widetilde{\mathcal{Q}} \cup \mathcal{Q}$;~
          read $\mulmat{i}$ from $\nextExpSets{}$
      \end{algosteps}
    \item \algoword{Return} $\monbas,\mulmat{1},\ldots,\mulmat{\nvars}$ }
\end{algobox}

The correctness of \cref{algo:multiplication_matrices} follows from the results
and discussions in this section. The next proposition implies the first item of
\cref{thm:mulmat}, and gives a more precise cost bound. It uses notation from
\cref{lem:nextvar_iteration}. We remark that one could easily verify that the
requirements of \cref{algo:multiplication_matrices} hold while building the
\(\ord\)-monomial basis \(\monbas\) at the first step, relying on the
characterization of \(\hyp{\genBy{\ltMod{\nodule}}}\) described in
\cref{lem:structural_assumption_mingens}.

\begin{proposition}
  \label{prop:algo:multiplication_matrices}
  Let $\ord$ be a monomial order on $\ring^\sdim$ and let $\gb$ be a reduced
  $\ord$-Gr\"obner basis such that \(\ring^\sdim/\nodule\) has dimension
  $\vsdim$, where $\nodule = \genBy{\gb}$. Assume
  \(\hyp{\genBy{\ltMod{\nodule}}}\), and using notation above, let \(\mu =
  \max(\mu_1,\ldots,\mu_{\nvars-1})\). Thus
  \[
    \mu \le \max \{\expnt \in \NN \mid \var_i^\expnt f \in \border \text{ for some } f \in \border \text{ and some } 2 \le i \le \nvars\},
  \]
  and in particular \(\mu \le \vsdim\). Then
  \cref{algo:multiplication_matrices} solves \cref{pbm:mulmat} using
  \begin{align*}
    & \bigO{\vsdim^\expmatmul \left(\nvars-1+\log(\mu_1\cdots\mu_{\nvars-1})\right) +
    \vsdim^{\expmatmul-1} \left(\card{\border-\LM} + \textstyle{\sum}_{1\le i\le \nvars} \sigma_i \log(\mu_i\nbvec_i/\sigma_i)\right)} \\
    & \subseteq\; \bigO{\nvars \vsdim^\expmatmul (1+\log(\mu))}
    \;\subseteq\; \bigO{\nvars\vsdim^\expmatmul \log(\vsdim)}
  \end{align*}
  operations in $\field$.
\end{proposition}
\begin{proof}
  First, Step 2 computes $\nextExpSets{\nvars} = \expSet_\nvars = \monbas \cup
  \LM$ and $\nf{\nextExpSets{\nvars}}$ from $\gb$, using $\bigO{\nvars
  \vsdim^2}$ computations of opposites of field elements (see
  \cref{lem:bivariate_mulmaty}). Then the \algoword{For} loop iteratively
  applies \cref{lem:nextvar_iteration} to obtain the remaining matrices. Using
  notation from \cref{lem:nextvar_iteration}, the overall cost bound is
  \begin{align*}
    & \bigO{\sum_{1\le i\le \nvars-1}\vsdim^{\expmatmul}(1+\log(\mu_i))
    + \vsdim^{\expmatmul-1}\sigma_i (1 + \log(\mu_i\nbvec_i/\sigma_i))} \\
    & \subseteq \bigO{\vsdim^{\expmatmul}(\nvars-1+\log(\mu_1\cdots\mu_{\nvars-1})) +
    \vsdim^{\expmatmul-1} \left(\card{\border-\LM} + \sum_{1\le i\le \nvars-1}\sigma_i\log(\mu_i\nbvec_i/\sigma_i)\right)}.
  \end{align*}
  Indeed, we have $\sigma_1 + \cdots + \sigma_{\nvars-1} = \card{\border-\LM}$,
  since $\sigma_i = \card{\expSet_i}$ and \(\expSet_1 \cup \cdots \cup
  \expSet_{\nvars-1} = \nextExpSets{1} - \expSet_\nvars = (\monbas\cup\border) -
  (\monbas\cup\LM) = \border-\LM\). Using the bounds $\card{\border-\LM} \le
  \nvars \vsdim$ (see \cref{sec:preliminaries}) as well as
  $\mu_i\nbvec_i/\sigma_i \le \mu_i \le \mu$ for all \(i\), we obtain the
  simplified cost bound \(\bigO{\nvars \vsdim^\expmatmul (1+\log(\mu))}\).
\end{proof}

\subsection{Change of order}
\label{sec:mulmat:change_order}

Combining the above algorithms leads to an efficient change of order procedure,
detailed in \cref{algo:change_order}.

\begin{algobox}
  \algoInfo
  {ChangeOrder}
  {algo:change_order}

  \dataInfos{Input}{%
    \item a monomial order $\ord_1$ on $\ring^\sdim$,
    \item a reduced $\ord_1$-Gr\"obner basis $\gb_1 \subseteq \ring^\sdim$,
    \item a monomial order $\ord_2$ on $\ring^\sdim$. }

  \dataInfos{Requirements}{%
    \item \(\ring^\sdim/\nodule\) has finite dimension \(\vsdim\) as a
      $\field$-vector space, where $\nodule = \genBy{\gb_1}$,
    \item \(\hyp{\genBy{\ltMod[\ord_1]{\nodule}}}\) holds. }

  \dataInfos{Output}{%
    \item the reduced \(\ord_2\)-Gr\"obner basis of \(\nodule\). }

  \algoSteps{
    \newcommand{\Indent}{\hphantom{char}}
    \newcommand{\DblIndent}{\hphantom{charchar}}
    \item \label{step:change_order:monbas} \(\monbas = (\basVec{1},\ldots,\basVec{\vsdim}), \mulmats=(\mulmat{1},\ldots,\mulmat{\nvars}) \assign \algoname{MultiplicationMatrices}(\ord_1,\gb_1)\)
    \item \label{step:change_order:evmat} \comment{Build a matrix \(\evMat\) such that \(\modRel = \nodule\)} \\
      \(\evMat \assign\) matrix in \(\matRing[\sdim][\vsdim]\) \\
      \algoword{For} \(1 \le i \le \sdim\): \\
      \Indent\algoword{If} \(\coordVec{i} = \basVec{j}\) for some \(1 \le j \le \vsdim\): \\
      \DblIndent \(i\)th row of \(\evMat \assign\) \([0\;\cdots\;0\;1\;0\;\cdots\;0] \in \matRing[1][\vsdim]\) with \(1\) at index \(j\) \\
      \Indent\algoword{Else}: \comment{in this case \(\coordVec{i} - \nf[\ord_1]{\coordVec{i}} \in \gb_1\)} \\
      \DblIndent \(i\)th row of \(\evMat \assign\) vector in \(\matRing[1][\vsdim]\) representing \(\nf[\ord_1]{\coordVec{i}}\) on the basis \(\monbas\)

    \item \label{step:change_order:gb} \(\gb_2 \assign \algoname{SyzygyModuleBasis}(\ord_2,\mulmats,\evMat)\)
    \item \algoword{Return} \(\gb_2\) }
\end{algobox}

As above concerning the computation of multiplication matrices, one may easily
verify from the input of \cref{algo:change_order} whether the requirements
hold. For simplicity, here we only use the simplified cost bounds of the above
results; better bounds may be obtained in particular cases.

\begin{proposition}
  \label{prop:algo:change_order}
  \cref{algo:change_order} is correct and uses
  \(\bigO{\sdim\vsdim^{\expmatmul-1} + \nvars \vsdim^\expmatmul \log(\vsdim)}\)
  operations in \(\field\).
\end{proposition}
\begin{proof}
  According to \cref{prop:algo:multiplication_matrices},
  \cref{step:change_order:monbas} uses
  \(\bigO{\nvars\vsdim^\expmatmul\log(\vsdim)}\) operations in \(\field\) and
  returns the \(\ord\)-monomial basis \(\monbas\) of \(\ring^\sdim/\nodule\)
  and the multiplication matrices \(\mulmats\) with respect to \(\monbas\). To
  build the matrix \(\evMat \in \matRing[\sdim][\vsdim]\),
  \cref{step:change_order:evmat} uses \(\bigO{\sdim\vsdim}\) operations;
  precisely, each normal form of a coordinate vector in the \algoword{Else}
  statement costs at most \(\vsdim\) computations of the opposite of an element
  in \(\field\). By \cref{prop:algo:syzygy_module_basis},
  \cref{step:change_order:gb} uses \(\bigO{\sdim\vsdim^{\expmatmul-1} + \nvars
  \vsdim^\expmatmul\log(\vsdim)}\) operations to compute the reduced
  \(\ord_2\)-Gr\"obner basis \(\gb_2\) of \(\modRel\). Hence the overall cost
  bound. Proving correctness amounts to showing that \(\modRel = \nodule\),
  which directly follows from the construction of \(\evMat\) and the fact that
  \(\coordVec{i} - \nf[\ord_1]{\coordVec{i}}\) is in \(\nodule\):
  \begin{align*}
    (p_1,\ldots,p_\sdim) \in \modRel
    & \;\Leftrightarrow\; \sum_{\substack{1\le i\le\sdim \\ \coordVec{i} \in \monbas}} p_i \coordVec{i}
                        + \sum_{\substack{1\le i\le\sdim \\ \coordVec{i} \not\in \monbas}} p_i \nf[\ord_1]{\coordVec{i}} \;\,\in \nodule \\
    & \;\Leftrightarrow\; \sum_{1\le i\le \sdim} p_i \coordVec{i} = (p_1,\ldots,p_\sdim) \in \nodule.
    \qedhere
  \end{align*}
\end{proof}

\bibliographystyle{elsarticle-harv}

\end{document}